%% file: writeup.tex
\documentclass[11pt]{article}
\usepackage{sn-preamble} 
\usepackage{tikz}
\usepackage{verbatim}
\usepackage{cancel}
\usepackage{caption}
\usepackage[subpreambles]{standalone}

\usetikzlibrary{decorations.pathreplacing,angles,quotes}
\usetikzlibrary{shapes.geometric}
\usetikzlibrary{patterns}

\usepackage{mathtools}
\usepackage{bm}

\newcommand{\ignore}[1]{}

\newcommand{\pmin}{p_{\mathrm{min}}}

\newcommand{\reldist}{\mathsf{rel}\text{-}\mathsf{dist}}

\newcommand{\fsm}{f_{\mathsf{sm}}}

\newcommand{\MQ}{\mathrm{MQ}}
\newcommand{\SAMP}{\mathrm{SAMP}}
\newcommand{\Sens}{\mathrm{Sens}}

\newcommand{\err}{\mathrm{err}}
\newcommand{\HermiteTest}{\textsc{Hermite-Test}}
\newcommand{\GSATest}{\textsc{GSA-Test}}
\newcommand{\CombinedTest}{\textsc{Combined-Test}}
\newcommand{\EstSense}{\textsc{Est-Sense}}
\newcommand{\MORStester}{\textsc{MORS-Tester}}
\newcommand{\Harmstester}{\textsc{Harms-Tester}}
\newcommand{\GSAFixedNoiseTest}{\textsc{GSA-Fixed-Noise-Test}}

\begin{document}

\title{
Sublinear-query relative-error testing of halfspaces
}

\author{
    \begin{tabular}{ccc}
        Xi Chen & Anindya De & Yizhi Huang \\
        Columbia University & University of Pennsylvania & Columbia University \\
        \url{xichen@cs.columbia.edu} & \url{anindyad@cis.upenn.edu} & \url{yizhi@cs.columbia.edu} \\
        \\
        Shivam Nadimpalli & Rocco A. Servedio & Tianqi Yang \\
        MIT & Columbia University & Columbia University \\
        \url{shivamn@mit.edu} & \url{rocco@cs.columbia.edu} & \url{tianqi@cs.columbia.edu}
    \end{tabular}
}

\date{
}

\pagenumbering{gobble}

\maketitle  

\begin{abstract}
The \emph{relative-error} property testing model was introduced in \cite{CDHLNSY24} to facilitate the study of property testing for ``sparse'' Boolean-valued functions, i.e.~ones for which only a small fraction of all input assignments satisfy the function. 
In this framework, the distance from the unknown target function $f$ that is being tested to a function $g$ is defined as $\Vol(f \hspace{0.06cm} \triangle \hspace{0.06cm}  g)/\Vol(f)$, where~the numerator is the fraction of inputs  on which $f$ and $g$ disagree and the denominator is the fraction of inputs  that satisfy $f$. 

Recent work \cite{CDHNSY26soda} has shown that over the Boolean domain $\{0,1\}^n$, any relative-error testing algorithm for the fundamental class of {halfspaces} (i.e.~linear threshold functions) must make $\Omega(\log n)$ oracle calls.  
In this paper we complement the \cite{CDHNSY26soda}  lower bound by showing that halfspaces can be relative-error tested over $\mathbb{R}^n$ under the standard $N(0,I_n)$ Gaussian distribution using a \emph{sublinear} number of oracle calls --- in particular, substantially fewer than would be required for learning.  Our results use a wide range of tools including Hermite analysis, Gaussian isoperimetric inequalities, and geometric results on noise sensitivity and surface area.

\end{abstract}

\newpage

\tableofcontents

\setcounter{page}{1}
\pagenumbering{arabic}

\input{sections/intro}

\input{sections/preliminaries-Gaussian}

\input{sections/Broccoli-tester}
\input{sections/MORS-tester}

\input{sections/combined-tester}

\section*{Acknowledgements}

This work was awarded a grant by the AI Security Institute (AISI) via the Alignment Project (``Discovering rare harmful behaviors exhibited by high-dimensional AI systems''). This work is also funded by OpenAI.
X.C.~is supported by NSF grants CCF-2106429 and CCF-2107187. 
A.D.~is supported by NSF grant CCF 2045128. 
Y.H.~is supported by NSF grants CCF-2211238, CCF-2106429, and CCF-2238221.
R.A.S.~is supported by NSF grants CCF-2211238 and CCF-2106429. 
T.Y.~is supported by NSF grants CCF-2211238, CCF-2106429, and AF-Medium 2212136.
T.Y.~and Y.H.~are also supported by an Amazon Research Award, Google CyberNYC award, and NSF grant CCF-2312242. 

\begin{flushleft}
\bibliographystyle{alpha}
\bibliography{allrefs}
\end{flushleft}

\appendix

\input{sections/appendix}

\end{document}

%% file: sections/intro.tex

\newpage

\section{Introduction} \label{sec:intro}

Over the past three decades Boolean function property testing has become a well-established research topic in theoretical computer science, see e.g.~the books and monographs \cite{Ron:10FNTTCS,Goldreich17book,BY22}. In this paper we consider an extension of the standard model of property testing, known as \emph{relative-error} testing, which was recently proposed in  \cite{CDHLNSY24}.

To describe the relative-error model, we first recall that in the standard model of Boolean function property testing the goal is to distinguish between the two cases that (i) $f \in {\cal C}$, where ${\cal C}$ is the class of functions that have the property that is being tested, versus (ii) $f$ is $\eps$-far (under the uniform distribution) from every function in ${\cal C}$, meaning that
\[
\dist(f,{\cal C}) \geq \eps, \text{~~~where~}
\dist(f,{\cal C}):=\min_{g \in {\cal C}} \dist(f,g) \text{~~~and~~~}
\dist(f,g) := \Prx_{\bx \sim \zo^n} [f(\bx) \neq g(\bx)].
\]
A testing algorithm in the standard model can obtain information about the  unknown function $f: \zo^n \to \zo$ that is being tested by making black-box oracle calls to $f$.  

The impetus for the study of relative-error testing is that the standard model described above is poorly suited for testing \emph{sparse} functions, i.e.~functions which have few satisfying assignments. This is because any such function has very small uniform-distribution distance to the constant-0 function, and so the tester can simply ignore the input function $f$ and answer according to the constant-0 function.  Indeed, for sparse functions it is arguably more natural to ask whether $f$ is ``close'' to having the property of interest (i.e.~close to belonging to ${\cal C}$) for a notion of closeness that is \emph{relative} to the small size of $f^{-1}(1)$ rather than to the ``absolute'' scale of all $2^n$ possible inputs. 

With this motivation, the relative-error property testing model which was defined in \cite{CDHLNSY24} changes the standard  model in the following ways:

\begin{flushleft}\begin{itemize}
\item 
The distance between the target function $f$ and a function $g$ is now measured using  \emph{relative distance}, which is  the uniform-distribution distance defined above scaled by the sparsity of $f$:
\[
\reldist(f,g) := {\frac {\Prx_{\bx \sim \zo^n}[f(\bx) \neq g(\bx)]}{\Prx_{\bx \sim \zo^n}[f(\bx)=1]}}
, \quad \text{i.e.}
\quad
\reldist(f,g) =  {\frac {\dist(f,g)}{\Pr_{\bx \sim \zo^n}[f(\bx)=1]}}.
\]
Relative distance thus captures the distance between $f$ and $g$ at the scale of $f$, and continues to be meaningful even if $f$ is very sparse.\footnote{It is easy to verify, as observed in \cite{CDHLNSY24}, that while $\reldist$ is not perfectly symmetric, if $\reldist(f,g) = \eps \leq 1/2$ then $\reldist(g,f)$ is also $\Theta(\eps)$, so $\reldist$ is symmetric up to constant factors in the setting we are interested in.
}

\item In addition to the usual  black-box oracle for $f$, a relative-error testing algorithm can also access a \emph{random sample oracle} $\SAMP(f)$, which takes no input and returns  a uniform random satisfying assignment $\bx \sim f^{-1}(1)$. (Note that without a $\SAMP(f)$ oracle, many black-box queries to $f$ could be required even to find a single satisfying assignment of $f$.)

\end{itemize}
\end{flushleft}

After the introduction of the relative-error testing model in \cite{CDHLNSY24}, the relative-error testability of a number of well-studied properties (or classes of functions) was investigated in~a number of papers \cite{CPPS25conjunctionDL,CPPS25junta,CDHNSY26soda,CPPPSZ25,CPPS26dnf}. Before turning to the specific results in those papers, we first give a quick overview of the general relationship between standard-model testability and relative-error testability.

\medskip
\noindent {\bf Standard-model versus relative-error testability.}   As shown already in \cite{CDHLNSY24}, standard-model testing is never more difficult than relative-error testing:  for any class ${\cal C}$, if ${\cal C}$ is relative-error testable to error $\eps$ using $q$ oracle calls then ${\cal C}$ is also standard-model testable to error $\eps$ using $O(q/\eps)$ oracle calls. 
(We remark that \cite{CPPS25conjunctionDL} gave a slight sharpening of this result under mild assumptions about the class ${\cal C}$.) 
On the other hand, \cite{CDHLNSY24}
  showed, by considering a contrived class of functions,
  that relative-error testing can sometimes be much more difficult (require many more queries) than standard-model testing.
This artificial class of functions leaves open the following natural question:  for ``natural'' classes of functions that are commonly studied in computational learning theory and concrete complexity, how does the query complexity of testing in the standard model compare to the relative-error model?

\medskip 
\noindent {\bf Relative-error testability of well-studied properties.}  The above question has been the subject of a significant amount of recent research \cite{CDHLNSY24,CPPS25conjunctionDL,CPPS25junta,CDHNSY26soda,CPPPSZ25,CPPS26dnf}.
Most of these results suggest that --- unlike the contrived example of \cite{CDHLNSY24} alluded to above --- relative-error testing often turns out to be essentially no  more (or at least not much more) difficult than standard-model testing for ``natural'' classes of functions ${\cal C}$.
In particular, the papers \cite{CDHLNSY24,CPPS25conjunctionDL,CPPS25junta,CPPPSZ25,CPPS26dnf} show that for the classes of monotone Boolean functions; 
unate Boolean functions; 
conjunctions; decision lists; $k$-juntas; subclasses of $k$-juntas such as size-$k$ decision trees and size-$k$ branching programs; and $s$-term DNF formulas, the query complexity of relative-error testing is  at most some fixed polynomial in the query complexity of standard-model testing.

However, an intriguing exception 
was discovered in the very recent work of \cite{CDHNSY26soda} on testing \emph{halfspaces}, also known as linear threshold functions or LTFs.\footnote{Recall that an $n$-variable halfspace is a function of the form $\sign(w \cdot x - \theta)$, 
with some $w\in \R^n$ and $\theta \in \R$.}
Halfspaces are one of the most intensively studied classes in computational learning theory, dating from the introduction of the Perceptron algorithm more than sixty years ago
\cite{Block:62,Novikoff:62}
down through to the present day \cite{KKMS:08, Daniely16, diakonikolas2020complexity, diakonikolas2024efficient,chandrasekaran2024smoothed}.
They have also been studied intensively both in property testing, see e.g.~\cite{GlasnerServedio:09toc,MORS10, MORS:09random, de2019your, de2021robust, Harms19, ChenPatel22}, and in other fields including probability theory and Boolean function analysis, see e.g.~\cite{MosselOdonnell:03, MosselNeeman15, MOO:10, Borell:85, DMN13, DFKO06}.
In \cite{MORS10} halfspaces were shown to be testable in the standard model with $\poly(1/\eps)$ queries independent of the ambient dimension $n$; in contrast, \cite{CDHNSY26soda} gave an $\tilde{\Omega}(\log n)$ lower bound on the number of queries that are required in the relative error model.

Given the \cite{CDHNSY26soda} lower bound, a natural goal --- which is the focus of the current paper --- is to gain a better understanding of the query complexity of testing halfspaces in the relative-error model.
In particular, the main question driving our research is the following:  
can $n$-dimensional halfspaces be relative-error tested with a number of queries that is \emph{sublinear} in $n$?  As motivation for this specific question, we remark that (as is well known) in many learning models the problem of \emph{learning} an unknown halfspace is known to require $\Theta(n)$ samples or queries. Since the complexity of learning is an upper bound on the complexity of testing (see \cite{GGR98}), it is natural to aim for $o(n)$-query testing algorithms, since this would show that testing is easier than learning, which is one of the central motifs and motivations for property testing of Boolean functions.  We further remark that the sample and query complexity of testing halfspaces in the \emph{distribution-free} model is known to be $\tilde{\Theta}(n)$ \cite{blais2021vc,ChenPatel22}, which lends additional impetus to the question of understanding whether it is possible to achieve sublinear complexity in the relative-error testing model.

\subsection{Our contributions:  Sublinear relative-error testing of halfspaces}  

Before describing our main results, we give some more context on halfspace testing in the standard model.  Halfspace testing has
primarily been studied under two  distributions, which correspond to discrete and continuous domains. These are the usual uniform distribution over $\{0,1\}^n$~\cite{MORS10, MORS:09random} and the standard  $N(0,I_n)$ Gaussian distribution over $\R^n$~\cite{MORS10, MosselNeeman15, de2019your, de2021robust}.
For $N(0,I_n)$, the ``standard-model'' distance between $f,g: \R^n \to \bits$ is simply
\[
\dist(f,g) := \Prx_{\bx \sim N(0,1)^n} [f(\bx) \neq g(\bx)].
\]
In the standard model \cite{MORS10} showed  that halfspaces can be tested using only $\poly(1/\epsilon)$ queries under both the uniform distribution over $\{0,1\}^n$ and the Gaussian  $N(0,I_n)$ distribution over $\R^n$. The \cite{MORS10} result for $\{0,1\}^n$ essentially uses their result over the Gaussian space~as a subroutine.\footnote{We remark that a very different tester for the Gaussian space than the \cite{MORS10} tester was given by Mossel and Neeman in \cite{MosselNeeman15}; these results will be discussed in more detail later.}
Thus, developing relative-error halfspace testing algorithms for the Gaussian distribution is a natural first step towards 
obtaining improved results for 
the uniform distribution over $\zo^n$.  

Towards this end, as the main results of this paper we provide three sublinear-query testing algorithms for halfspaces, all over the Gaussian space, where the relative distance between $f$ and another function $g$, where $f,g: \R^n \to \zo,$ is defined (analogous to the Boolean setting) as
\[
\reldist(f,g) := {\frac {\Prx_{\bx \sim N(0,I_n)}[f(\bx) \neq g(\bx)]}{\Prx_{\bx \sim N(0,I_n)}[f(\bx)=1]}}
, \quad \text{i.e.}
\quad
\reldist(f,g) =  {\frac {\dist(f,g)}{\Pr_{\bx\sim N(0,I_n)}[f(\bx)=1]}},
\]
and in this setting a call to the random sample oracle $\SAMP(f)$ returns a draw from 
the Gaussian distribution conditioned on $f^{-1}(1)$. (See   \Cref{sec:preliminaries} for a more detailed definition of the model.)

\begin{remark} \label{rem:easyp}
Before detailing our algorithmic results, we remark that by combining the $\poly(1/\eps)$-query testing algorithms of \cite{MORS10,MosselNeeman15} for halfspaces under the Gaussian $N(0,I_n)$ distribution in the \mbox{standard} model  and Fact~9 of \cite{CDHLNSY24}, we easily obtain relative-error \mbox{testing} algorithms with query complexity $\poly(1/(\eps p))$ when we are promised that the ``Gaussian volume'' $\Pr_{N(0,I_n)}[f(\bx)=1]$ of the function $f$ being tested is at least $p$. (Intuitively, this is because achieving relative error $\eps$ for a ``$p$-volume'' function is the same as achieving standard error $\eps p$.)  However, this is a very poor bound in the (most interesting) case when $p$ is very small; in the results below, we achieve an exponential improvement by giving algorithms whose dependence on $n$ is sublinear and whose dependence on $p$ is only inverse \emph{poly-logarithmic}.
\end{remark}

\noindent {\bf Our results.} For our first two results, the algorithm is assumed to be given (an accurate estimate of) the ``Gaussian volume'' $p:=\Pr_{N(0,1)^n}[f(\bx)=1]$ of the unknown target function $f$, 
which we sometimes refer to as the \emph{volume} of $f$ and denote by $\Vol(f)$. In the third result, the algorithm is only given a lower bound $p_{\text{min}}$ on $\Vol(f)$ rather than an accurate estimate.
We first give informal statements of our three algorithms' guarantees and provide detailed theorem statements later.
\begin{flushleft}\begin{enumerate}
    \item The first algorithm (see \Cref{thm:Gaussian1}) uses $\poly(\log (1/p), 1/\epsilon)$ samples and queries to do $\eps$-relative-error testing. At a high level, similar to \cite{MosselNeeman15}, the algorithm relies on the robust extremal isoperimetric properties of halfspaces over the Gaussian space. 
See \Cref{sec:broccoli-bound-techniques} for a more detailed description of the ideas underlying this algorithm.
    
    \item The second algorithm (see \Cref{thm:Gaussian2}) uses 
$\tilde{O}(\sqrt{n}) \cdot \poly(\log(1/p), 1/\eps)$ samples to do $\eps$-relative-error testing. In contrast with the first one, this algorithm does not use any black-box queries to $f$, only random samples drawn from $\SAMP(f)$. 
The high-level idea behind this tester is that for halfspaces over the Gaussian space, the level-$1$ Hermite weight depends only on the bias of the function, and this relationship is robust enough to serve as the basis for a testing algorithm (\Cref{sec:MORS-techniques} gives a more complete overview).   
 We note that similar ideas were used in the original work of \cite{MORS10} for testing under the Gaussian distribution, as well as a number of other related works including \cite{Harms19} and \cite{de2023testing}.

\item The third algorithm (see \Cref{thm:Gaussian3}) addresses the problem  when $p$ is unknown; it only requires a lower bound $p_{\min}$ on the value of $p$ rather than an accurate estimate. It uses 
$O(\sqrt{n}) \cdot \poly(\log(1/p_{\min}), 1/\eps)$ samples and queries to do $\eps$-relative-error testing.
 The high level idea is to use the estimator from the proof of \Cref{thm:Gaussian2} to obtain an \emph{upper} bound on $p$, in conjunction with a variant of \Cref{thm:Gaussian1} which only errs if it is given an estimate of $p$ which is \emph{too low}. See \Cref{sec:Gaussian3-techniques} for a more detailed overview of the main ideas underlying this algorithm.
\end{enumerate}\end{flushleft}

Here are detailed statements of our algorithmic results:

\begin{theorem}
[Gaussian halfspace testing for known $p$] \label{thm:Gaussian1}
There is an algorithm \GSATest~with the following property:  If \GSATest~is given~sample access $\SAMP(f)$ and black-box access $\MQ(f)$ to a measurable function $f: \R^n \to \{0,1\}$, a  parameter $\eps$, and~an~estimate $\hat{p}$ of $p := \Vol(f)$ satisfying 
\[
{\frac {p}{1+\zeta}} \leq \hat{p} \leq (1+\zeta) p,\quad\text{where}\quad \zeta = \frac{c_1 \epsilon^2}{\log^2(1/p)}
\]
for a suitably small absolute constant $c_1 > 0$,
then it makes
$O(\epsilon^{-14} \log^{13}(1/p))$  calls  to the oracles to 
  test whether $f$ is a halfspace or $\eps$-far from any halfspace in
  relative distance.
\end{theorem}

\begin{theorem}
[Sample-based Gaussian halfspace testing for known $p$] \label{thm:Gaussian2}
There is a sample-based algo\-rithm \HermiteTest\ with the following property:  If it  is given sample access~$\SAMP(f)$ to a measurable function $f: \R^n \to \{0,1\}$, a parameter $\eps$, and an estimate $\hat{p}$ of $p := \Vol(f)$ satisfying 
\[
{\frac {p}{1+\eta}} \leq \hat{p} \leq (1+\eta) p,
\quad
\text{where}\quad \eta = \frac{c_2\eps^2}{\log(1/p)}
\]
for some suitably small absolute constant $c_2>0$,
then \HermiteTest\  makes 
$$\max\cbra{
\Theta\pbra{
    {\frac {\sqrt{n}}{\eps^2}} + {\frac {\log^2(1/p)}{\eps^4}}
    },\tilde{\Theta}\left(\frac{\sqrt{n}}{\eps^{7}}\right)
    }
$$ calls to $\SAMP(f)$ to test whether $f$ is a halfspace
  or $\eps$-far from any halfspace in relative distance.
\end{theorem}

We remark that the requirement in \Cref{thm:Gaussian2} that \HermiteTest\ is given an estimate of $p$ is inherent to any sample-based algorithm rather than a limitation of our arguments. In \Cref{ap:sample-only-lower-bound} we observe that any algorithm which uses only samples from $\SAMP(f)$ and is \emph{not} given an estimate of $p$ must draw $\Omega(n)$ samples,
even if we are given a promise that $p$ is either $1/2$ or 1.

\begin{theorem}
[Gaussian halfspace testing for unknown $p$] \label{thm:Gaussian3}
There is an algorithm \CombinedTest\ with the following property: 
If it is given sample access $\SAMP(f)$ and black-box access $\MQ(f)$ to an unknown measurable function $f: \R^n \to \{0,1\}$, which is guaranteed to satisfy $p := \Vol(f) \geq \pmin$,
 and a parameter $\eps$,  
then it makes 
\[
\Theta
\pbra{
    {\frac {\sqrt{n \cdot{\log (1/\pmin)}}}{\eps^2}} +  {{\frac {\log^6(1/\pmin)}{\eps^{12}}}
    }}
\]
calls to the oracles to test whether $f$ is a halfspace or $\eps$-far from any halfspace in relative distance.
\end{theorem}

\section{Technical overview}
\label{subsec:technical-overview}

\subsection{Our algorithm for halfspace testing  using queries and samples:\\ Overview of \Cref{thm:Gaussian1}} \label{sec:broccoli-bound-techniques}

To explain the main idea of \Cref{thm:Gaussian1}, we start by recalling
two fundamental notions in Gaussian analysis. For any set $A \subseteq \mathbb{R}^n$, its \emph{Gaussian volume}, denoted by $\Vol(A)$, is $\Pr_{\bx \sim N(0,I_n)}[\bx \in A]$.
Furthermore, if $A$ is sufficiently smooth, we can also define its \emph{Gaussian surface area}, denoted $\mathsf{surf}(A)$, as $\int_{x \in \partial A} \varphi_n(x) d\sigma(x)$, where $\varphi_n(\cdot)$ denotes the standard $n$-dimensional Gaussian density; in other words, we integrate the standard Gaussian density   over the surface of $A$. Given these notions of volume and surface area over the Gaussian space, one is naturally led to the ``isoperimetric problem:"  {\em for a given (Gaussian) volume, what set minimizes the (Gaussian) surface area?}

This question was answered by Borell~\cite{Borell:75} and Sudakov-Tsirelson~\cite{ST:78}, who showed that for any $0 < p <1$, for any measurable set $A \subseteq \mathbb{R}^n$ with $\Vol(A)=p$, it always holds that $\mathsf{surf}(A) \ge \mathsf{surf}(H_p)$ where $H_p$ is a halfspace with $\Vol(H_p)=p$. In other words, in the Gaussian space, for any given volume $p$, the halfspace with volume $p$ has minimal surface area. Subsequently, Ehrhard~\cite{ehrhard1986elements} showed that halfspaces are the \emph{unique} minimizers of the Gaussian surface area. In other words, if a set $A \subseteq \mathbb{R}^n$ has $\Vol(A)=p$ and $\mathsf{surf}(A) =\mathsf{surf}(H_p)$, then the set $A$ is a halfspace. 

One is now naturally led to the {\em robustness} question: Suppose  a set $A$ has $\Vol(A)=p$ and $\mathsf{surf}(A)$ is $\delta$-close to $\mathsf{surf}(H_p)$. Then is it the case that $A$ is $\delta'$-close to a halfspace, where $\delta' \rightarrow 0$ as $\delta \rightarrow 0$? Cianchi~{\em et~al.}~\cite{cianchi2011isoperimetric} were the first to affirmatively answer this question. Soon thereafter, in a significant breakthrough, Mossel and Neeman~\cite{mossel2015robust} obtained  a ``dimension free robust isoperimetry result", i.e.~one in which the dependence between $\delta'$ and $\delta$ is independent of the ambient dimension $n$. 

This result was quantitatively improved by Eldan~\cite{eldan2015two} and then by Barchiesi~{\em et~al.}~\cite{barchiesi2017sharp}
(the latter result is the underlying technical ingredient we use in \Cref{thm:Gaussian1}). 
A natural question at this point is ``what is the connection between robust isoperimetry  and testing halfspaces?" 

The connection arises from a 2014 result of Neeman~\cite{Neeman14}, who gave an  algorithm for testing surface area over Gaussian space. The precise technical guarantee is somewhat cumbersome to state, but roughly speaking, the algorithm makes $\poly(S, 1/\epsilon)$-queries to an unknown set $A$ and (i) accepts with high probability if $\mathsf{surf}(A) \leq S$, but (ii) rejects with high probability if $\mathsf{surf}(B) > (1 + \epsilon)S$ for every set $B$ which is $\epsilon$-close to $A$\footnote{This perturbation of $A$ is necessary because one can increase the surface area of a set arbitrarily by modifying on a measure zero set.}. Now, observe that if an unknown set $A$ of some given volume $p$ is a halfspace, then its surface area is exactly given by $\varphi(\Phi^{-1}(p))$ where $\Phi$ is the cdf and $\varphi$ is the pdf of the standard Gaussian.  So to test whether $A$ is a halfspace, we can use Neeman's algorithm to test whether its surface area is essentially that of a halfspace of volume $p$, and accept if and only if that is the case.  Correctness of this simple algorithm is established using the dimension-free robust isoperimetry result of Barchiesi~{\em et~al.}~\cite{barchiesi2017sharp} that was mentioned above.

\subsection{Our algorithm for sample-based halfspace testing:  Overview of~\Cref{thm:Gaussian2}} \label{sec:MORS-techniques}

Our starting point for~\Cref{thm:Gaussian2} is the work of Matulef, O'Donnell, Rubinfeld, and Servedio~\cite{MORS10} which gives an algorithm (the ``MORS algorithm'') that tests halfspaces over $\R^n$ using $O_\eps(1)$ queries. The MORS algorithm relies on two structural ingredients: 
\begin{itemize}
    \item First, there is a function $U: [0,1]\to\R$ (see~\Cref{def:U}) which, given as input the volume $\Vol(f)$ of a function $f:\R^n\to\zo$, tells us exactly what the \emph{level-$1$ Hermite weight} 
    \[
        \bW^{=1}[f] := \sumi \Ex_{\bx\sim N(0,I_n)}\sbra{f(\bx)\cdot\bx_i}^2
    \]
    would be if the function $f$ were the indicator function of a halfspace.
    \item Second, if $U(\vol(f)) \approx \bW^{=1}[f]$, then the set $f$ is close to some halfspace.
    \end{itemize}
The above facts suggest a natural algorithm: estimate $\vol(f)$ and $\bW^{=1}[f]$, and then compare $U(\vol(f))$ to $\bW^{=1}[f]$. 
(We remark that the MORS algorithm uses queries to do this.)

The algorithm~\HermiteTest~(\Cref{alg:MORS}) that yields~\Cref{thm:Gaussian2} can be viewed as a {relative-error, sample-based} variant of the MORS algorithm. Our main technical lemma (\Cref{lem:our-MORS}) shows that if $U(\vol(f))$ is sufficiently close to $\bW^{=1}[f]$ in a suitable sense, then $f$ is close \emph{under relative distance} to some halfspace. In this sense,~\Cref{lem:our-MORS} can be viewed as a relative-error strengthening of the second structural ingredient of the MORS algorithm above. 
With \Cref{lem:our-MORS} in hand, we apply a test which is similar in spirit to the  MORS algorithm: estimate $\bW^{=1}[f]$ and see if it is (sufficiently) close to $U(\vol(f))$. 

This leads to a second technical challenge: in order to get a sample-based algorithm we must accurately estimate $\bW^{=1}[f]$ using access to samples from $\SAMP(f)$ alone. (Note that the MORS algorithm relies on query access to the function $f$, and there are sample-based testers~\cite{BBBY12,Harms19} that rely on \emph{labeled} samples from $f$ as opposed to just positive samples from $\SAMP(f)$.) To accomplish this, we employ a recent estimator used in a different context by~\cite{de2023testing,de2024detecting} (for the problem of detecting distribution truncation) which is as follows: given a draw of $2m$ i.i.d. samples $\x{1},\dots,\x{m}, \y{1},\ldots,\y{m} \leftarrow \SAMP(f)$, we compute the statistic
\[
    \bT := \frac{1}{m^2}\sum_{i\neq j} \x{i}\cdot\y{j}.
\]
It is readily verified that $\E[\bT] = \vol(f)^{-2}\cdot\bW^{=1}[f]$; since $\Vol(f)$ is known, we can estimate $\bW^{=1}[f]$ if we can estimate $\E[\bT]$. To estimate $\E[\bT]$, we must establish a bound on the variance of $\bT$; to do this, we employ (consequences of) \emph{hypercontractivity of the Gaussian measure}~(see~Chapters 9 and 11 of~\cite{odonnell-book}), drawing inspiration from 
recent work~\cite{de2024detecting}.

\subsection{Our algorithm for halfspace testing  when the volume is unknown:\\ Overview of \Cref{thm:Gaussian3}} \label{sec:Gaussian3-techniques}
We obtain the algorithm of \Cref{thm:Gaussian3}, which does not need to have (an estimate of) the value of $p$, by combining the ideas underlying \Cref{thm:Gaussian1} and \Cref{thm:Gaussian2}, as well as bringing in new ideas.  In particular, a crucial structural property of LTFs which we use for \Cref{thm:Gaussian3} is the following:  not only does the level-$1$ Hermite weight of an LTF $f$ depend only on the bias of the function $\E[f]$ as mentioned earlier, but under the Gaussian distribution LTFs are in fact \emph{robustly extremal}, among all $\zo$-valued functions, in terms of how their level-1 Hermite weight depends on  on the bias of the function. 
In more detail, writing $h_p$ to denote a halfspace with Gaussian volume $p$, it is the case that \emph{any} function $f: \R^n \to \zo$ with Gaussian volume $p$ must satisfy
\begin{equation} \label{eq:extremal}
    \bW^{=1}[f] \leq     \bW^{=1}[h_p].
\end{equation}
(As alluded to in the last paragraph of \Cref{sec:MORS-techniques}, the LHS of the inequality above is precisely $\Vol(f)^2 \cdot \Ex_{\bx,\by \sim f^{-1}(1)}[\bx \cdot \by].$)
Moreover the inequality can only be an equality if $f$ is a halfspace, and it can only be close to an equality if $f$ is close to a halfspace.  Versions of this structural property are established, in varying degrees of explicitness, in 
Theorem~7 of \cite{KKMO07}, Theorem~26 of \cite{MORS10}, and Theorem~3.3.4 of 
\cite{Harmsthesis}; our \Cref{lem:our-MORS2} gives a version of this which is --- crucially, for our purposes --- quantitatively stronger than those earlier results in the small-$p$ regime.

With this perspective, we can reinterpret the testing algorithm of \Cref{thm:Gaussian2} in the following light: if the algorithm were \emph{not} given $p$, it could compute an estimate $\tau$ of $\Ex_{\bx,\by \sim f^{-1}(1)}[\bx \cdot \by]$ using draws from $\SAMP(f)$ and solve the equation 
\[
\tau = {\frac {\bW^{=1}[h_{p_2}]} {(p_2)^2}}
\]
for the value $p_2$. If $f$ were an LTF and the value of $\tau$ were precisely the value of $\Ex_{\bx,\by \sim f^{-1}(1)}[\bx \cdot \by]$, then this would give us the exact correct value $p_2 = p$. On the other hand, because of the extremal property \eqref{eq:extremal} mentioned above, if $f$ is far from every LTF then (up to a small factor due to estimation error) the value of $p_2$ will be \emph{significantly larger than}  the true value $p = \Vol(f).$

Now let us return to \Cref{thm:Gaussian1}; the key for our analysis is that this result relies on a robust extremal property of LTFs that \emph{goes in the other direction}.  
Recall that the \Cref{thm:Gaussian1} algorithm essentially works by estimating 
\[
\Prx_{\bx \sim f^{-1}(1),\by \sim N_\lambda(\bx)}[f(\bx)\neq f(\by)], \quad \text{i.e.}\quad 
\Prx_{\bx \sim f^{-1}(1),\by \sim N_\lambda(\bx)}[f(\by)=0]
\]
(where ``$\by \sim N_\lambda(\bx)$'' means that $\by$ is a $\lambda$-correlated random perturbation of $\bx$; see \Cref{subsec:cherry-coke-zero} for a formal definition),
and checking whether it takes the ``right value" that it should take for an 0/1-valued LTF with expectation $p$, which is ${\frac 1 {2p}} \Sens_\lambda(LTF_p) = C\sqrt{\lambda \ln(1/p)}.$  In the \Cref{thm:Gaussian1} algorithm we use a value of $\lambda$ which depends on $p$, but for a \emph{fixed} value of $\lambda$, the function ${\frac 1 {2p}} \Sens_\lambda(LTF_p)$ is a \emph{decreasing} function of $p$; as we will see, this will be crucial for us.

Now, known isoperimetric properties of halfspaces can be shown to imply that if $f$ is any $\zo$-valued function $f$ with $\Vol(f)=p$, then $f$ must (essentially) satisfy 
\[
\Prx_{\bx \sim f^{-1}(1),\by \sim N_\lambda(\bx)}[f(\bx)\neq f(\by)] \geq {\frac 1 {2p}} \Sens_\lambda(h_p);
\]
 i.e.~that halfspaces are  \emph{minimizers} of the quantity $\Prx_{\bx \sim f^{-1}(1),\by \sim N_\lambda(\bx)}[f(\bx)\neq f(\by)]$, which the algorithm of \Cref{thm:Gaussian1} estimates from samples and queries.
 This means that if we run the algorithm of \Cref{thm:Gaussian1} with the value $p_2$ obtained from the modified algorithm of \Cref{thm:Gaussian2} as described earlier, there are two possibilities:
 
 \begin{itemize}
 \item If $f$ is a halfspace, then $p_2$ will be (approximately) equal to the true value of $p$, and since we are running the algorithm of \Cref{thm:Gaussian1} with (essentially) the right value of $p$, it will accept.
 
 \item On the other hand, if $f$ is far from every halfspace, then as explained earlier $p_2$ will be significantly \emph{larger} than the true value of $p$.  But since ${\frac 1 {2p}} \Sens_\lambda(LTF_p)$ is a decreasing function of $p$, the algorithm of \Cref{thm:Gaussian1} can only accept if it is run on a ``guessed'' value of $\Vol(f)$ which is either \emph{equal to or smaller than} the true value of $\Vol(f)$. So in this case, the algorithm will reject.
\end{itemize}
This concludes the overview of the proof of \Cref{thm:Gaussian3}. One point of technical interest is that the actual proof of \Cref{thm:Gaussian3} does not use \Cref{thm:Gaussian1} as a black-box. In fact, unlike \Cref{thm:Gaussian1} which crucially relies on the robust isoperimetry results of \cite{barchiesi2017sharp}, \Cref{thm:Gaussian3} only relies on the classical result of Borell~\cite{Borell:85} which states that halfspaces are the minimizers of noise sensitivity. The result of Borell has many proofs by now, including some quite elementary (see~\cite{DMN13} and references therein). So, arguably, the proof of \Cref{thm:Gaussian3} is significantly simpler than that of \Cref{thm:Gaussian1} (though conceptually, \Cref{thm:Gaussian3} is inspired by \Cref{thm:Gaussian1}).

%% file: sections/preliminaries-Gaussian.tex

\section{Preliminaries for relative-error testing over $N(0,1)^n$} \label{sec:preliminaries}

We use boldfaced letters such as $\bx, \boldf,\bA$, etc. to denote random variables (which may be \mbox{real-valued,} vector-valued, function-valued, or set-valued; the intended type of the random variable will be clear from the context).
We write $\bx \sim \calD$ to indicate that the random variable $\bx$ is distributed according to probability distribution $\calD$.

We will write $(e_i)_{i=1}^n$ for the collection of standard basis vectors in $\R^n$. Given two sets $A$ and $B$, we use $A \, \triangle \, B$ to denote their symmetric difference, i.e. $A\,\triangle\,B = (A\setminus B) \cup (B\setminus A)$. 
As in \cite{odonnell-book} we write $\mathbb{N}$ for the set $\mathbb{N}=\{0,1,2,\dots\}$.\medskip

\medskip
\noindent {\bf The Gaussian distribution and Gaussian relative-error testing.}
We will write $N(0,I_n)$ to denote the $n$-dimensional standard Gaussian distribution, and denote its density function by $\phi_n$, i.e., for $x\in \mathbb{R}^n$,
\[\varphi_n(x) = \frac{1}{\sqrt{(2\pi)^{n}}} \cdot e^{-\|x\|^2/2}.\]
We frequently write $\phi \equiv \phi_1$ to denote the 
  one-dimensional standard Gaussian density. 
We write $\vol(K)$ to denote the Gaussian measure of a (Lebesgue measurable) set $K \subseteq \R^n$, that is 
\[\vol(K) := \Prx_{\bx \sim N(0,I_n)}[\bx \in K].\]  
Given a $0/1$-valued function $f:\R^n\to\zo$, we write $\Vol(f)$ to denote $\Vol(f^{-1}(1))$.

Finally, we will write $\Phi: \R\to(0,1)$ for the c.d.f.~of $N(0, 1)$, i.e. 
\[\Phi(r) = \Prx_{\bg\sim N(0, 1)}\sbra{\bg\leq r} = \int_{-\infty}^r \phi(x)\,dx.\]

We will frequently use the following standard tail bound on (univariate) Gaussian random variables:

\begin{proposition}[Proposition~2.1.2 of \cite{vershynin2018high} or Exercise~2.2 of \cite{Wainwright19book}] \label{prop:gaussian-tails}
	Suppose $\bg\sim N(0,1)$ is a one-dimensional Gaussian random variable. Then {for all $r>0$,} we have
\[
    \left({\frac 1 r} - {\frac 1 {r^3}} \right) \cdot \varphi(r)
    \leq  \Phi(-r)=\Prx_{\bg \sim N(0,1)}[\bg \geq r] \leq
    \pbra{{\frac 1 r} - {\frac 1 {r^3} + {\frac 3 {r^5}}}}\cdot \varphi(r).
\]
\end{proposition}

From \Cref{prop:gaussian-tails} we get the following estimate on  $\Phi(\cdot)$ that will be convenient:
\begin{fact}~\label{fact:Gaussian-cdf-estimate}
Let $r$ be any positive value bounded away from $0$ by an absolute constant, e.g.~$r \ge 0.01$. 
Then we have 
$
1-\Phi(r)=\Phi(-r) =  \Theta(e^{-r^2/2}/r). 
$
\end{fact}

\medskip
\noindent {\bf Gaussian relative-error testing.}
The model of relative-error testing under the Gaussian distribution is a natural variant of the model of relative-error testing over $\zo^n$ that was described earlier.
Now a call to $\SAMP(f)$ returns a draw from $N(0,1)^n|_{f^{-1}(1)}$, the standard Gaussian distribution conditioned on $f^{-1}(1)$.
Similar to before, a relative-error testing algorithm for ${\cal C}$ must output ``yes'' with high probability (say at least 9/10; this success probability can be easily amplified) if $f \in {\cal C}$, and must output ``no'' with high probability (again, say at least 9/10) if $\reldist(f,{\cal C}) \geq \eps$, where 
$\reldist(f,{\cal C})=\min_{g \in {\cal C}}\reldist(f,g)$ and 
  the relative distance between $f$ and $g$ is defined as
\[
\reldist(f,g) = {\frac {\Vol(f^{-1}(1) \hspace{0.06cm}\triangle\hspace{0.06cm} g^{-1}(1))}{\Vol(f)}}.
\]

\medskip
\noindent
{\bf Testing algorithms in the standard model.} For handling certain edge cases, we will rely on known testing algorithms in the standard (not relative-error) property testing model. One such result that we will use is the following theorem due to \cite{MORS10}:  
\begin{theorem}
[Theorem~26 of \cite{MORS10}] \label{thm:mainMORS}
There is an algorithm \MORStester~that, when given error parameter $\epsilon>0$ and oracle access to $f: \mathbb{R}^n \rightarrow \zo$, makes $O(1/\eps^{12})$
queries to $f$ and has the following guarantee: 
\begin{flushleft}\begin{enumerate}
    \item If $f$ is a halfspace, the algorithm outputs \textsf{accept} with probability at least $0.99$. 
    \item If $\Vol(f^{-1}(1) \ \triangle \ h^{-1}(1)) \ge \eps$ for every halfspace $h$, then the algorithm outputs \textsf{reject} with probability at least $0.99$. 
\end{enumerate}\end{flushleft}
\end{theorem}

For sample-based testing, we will use the following result due to Harms \cite{Harms19}:
\begin{theorem}
[Theorem~1.1 of \cite{Harms19}] \label{thm:mainHarms}
There is an algorithm \Harmstester~that, when given error parameter $\epsilon>0$ and access to random labeled examples $(\bx,f(\bx))$ with each $\bx \sim N(0,I_n)$, uses $\tilde{O}(\sqrt{n}/\eps^7)$ samples and has the following guarantee: 
\begin{flushleft}\begin{enumerate}
    \item If $f$ is a halfspace, the algorithm outputs \textsf{accept} with probability at least $0.99$. 
    \item If $\Vol(f^{-1}(1) \ \triangle \ h^{-1}(1))\ge \eps$ for every halfspace $h$, then the algorithm outputs \textsf{reject} with probability at least $0.99$. 
\end{enumerate}\end{flushleft}
\end{theorem}

%% file: sections/Broccoli-tester.tex

\section{Testing with known volume via Gaussian surface area: \\Proof of \Cref{thm:Gaussian1}}
\label{sec:broccoli}

In this section we prove \Cref{thm:Gaussian1}. 
Before we start, note that we can make a simplifying assumption, which is that $p :=\Vol(f)$ is at most $0.1$. This is because if $p>0.1$, we can just run the testing algorithm from \cite{MORS10} (\Cref{thm:mainMORS}) with error parameter $\epsilon / 10$; as $p>0.1$, standard-model $\eps/10$-testing implies $\eps$-relative-error testing. So for the rest of this section, we assume that $p \leq 0.1$. 

\subsection{Gaussian surface area}
\label{subsec:gsa}

We begin by recalling the notion of surface area in the Gaussian space. 
\begin{definition}~\label{def:surface-area}
Let $A \subseteq \mathbb{R}^n$. The surface area of $A$ (under the standard Gaussian measure $\phi_n$) is given by 
\[
\mathsf{surf}(A) := \lim_{\delta \rightarrow 0^+} \frac{\Vol(A_\delta) - \Vol(A)}{\delta}, 
\]
where $A_\delta :=\{x : d(x,A) \le \delta\}$. For sets $A$ with a smooth boundary, $\mathsf{surf}(A)$ is equivalent to 
\[
\int_{x \in \partial A} \phi_n(x) d \sigma(x), 
\]
where $d\sigma(x)$ is the standard surface area element in $\mathbb{R}^n$ and $\partial A$ denotes the boundary of the set $A$.
\end{definition}

There is a close connection between the notion of surface area and noise sensitivity, as shown by Ledoux~\cite{Ledoux:94}: 
\begin{theorem}\label{thm:noise-sensitivity-Ledoux}
For any $t \ge 0$ and any set $A \subseteq \mathbb{R}^n$ with ${\cal C}^1$ boundary,
we have 
\[
\NS_t(A) \le \frac{2 \sqrt{t}}{\sqrt{\pi}} \cdot \mathsf{surf}(A). 
\]
\end{theorem}
A near-converse of this result was shown by Neeman~\cite{Neeman14}:
every set $A$ as above is close to a set $B$ whose surface area can be bounded in terms of the noise sensitivity of $A$.
\begin{theorem} [Theorem 1.2, Theorem 2.1 of \cite{Neeman14}] 
\label{thm:Neeman-reverse-Ledoux}
    Let $A \subseteq  \mathbb{R}^n$ be a set with ${\cal C}^1$ boundary and let $t, \xi >0$.  
     Then, there exists a set $B \subseteq \mathbb{R}^n$  such that \begin{enumerate}
        \item $\displaystyle\Vol(A \ \triangle \ B) \le \frac{\NS_t(A)}{\xi}$;\quad\text{and} 
        \item $\displaystyle \mathsf{surf}(B)  \leq  \sqrt{\frac{\pi}{2}} \cdot \bigg(1 + O\bigg(\frac{\xi}{\sqrt{\log(1/\xi)}}\bigg)\bigg)\cdot   \frac{1}{\sqrt{e^{2t}-1}} \cdot \NS_t(A)$.  
    \end{enumerate}
\end{theorem}

We recall (see e.g.~Chapter~5 of~\cite{odonnell-book}) that the function 
\begin{equation} \label{eq:I}
I(p) := \varphi(\Phi^{-1}(p))
\end{equation}
is sometimes known as the \emph{Gaussian isoperimetric function} (cf.~\Cref{fact:halfspace-sa} below).  This function will be used in various ways throughout our arguments, including in the proof of the following simple fact:

\begin{fact}~\label{fact:important-Gcdf}
Let $0 < q \leq p \leq 1/2$ and $p / q \le (1+\epsilon)$. Then, 
\[
\varphi(\Phi^{-1}(p)) \ge \varphi(\Phi^{-1}(q)) \ge \frac{1}{(1+\epsilon)} \varphi(\Phi^{-1}(p)).
\]
\end{fact}
\begin{proof}
The function $I(\cdot)$ is easily seen to be monotonically non-decreasing in the interval $(0,1/2]$, which gives the first inequality.
Moreover, it is known that the function $I(p)$ is concave (see Exercise~5.24~\cite{odonnell-book}). This implies that 
\[
I(q) \ge \frac{q}{p} I(p)\ge \frac{1}{1+\epsilon} I(p). 
\qedhere
\]
\end{proof}

We will also need some basic facts about the function $\varphi(\Phi^{-1}(p))/p$ which we state here. 
\begin{claim}~\label{claim:facts-about-Ipp}
The function $R(p):= \frac{\varphi(\Phi^{-1}(p))}{p}$ is decreasing in the interval $(0,1/2)$.
\end{claim}
\begin{proof}
Let $s \in (-\infty,0)$ be chosen so that $\int_{-\infty}^s \varphi(t) dt = p$. Then, note that 
\[
R(p) = \frac{\varphi(s)}{\int_{-\infty}^s \varphi(t) dt}.
\qedhere
\]
\end{proof}

We also need the following simple fact characterizing the surface area of any halfspace in the Gaussian space in terms of its volume; the proof of the final equality is a calculation using standard bounds on the Gaussian pdf and cdf, in particular \Cref{prop:gaussian-tails}.
\begin{fact}~\label{fact:halfspace-sa}
Let $f: \mathbb{R}^n \rightarrow \{0,1\}$ be a halfspace and let $\Vol(f^{-1}(1))=p$. Then
\[
\mathsf{surf}(f^{-1}(1))  =  \varphi(\Phi^{-1}(p)) = \Theta\left(p \sqrt{\ln(1/\min\{p, 1-p\})}\right).
\]
\end{fact}

\subsection{Other basic algorithmic and structural ingredients}
\label{subsec:cherry-coke-zero}

\noindent {\bf Gaussian noise sensitivity.} For any $x \in \mathbb{R}^n$ and any $t \ge 0$, we let $N_t(x)$ denote the distribution over $\mathbb{R}^n$ defined by $e^{-t} x + \sqrt{1-e^{-2t}} \by$ where $\by\sim N(0,I_n)$.
Further, for any function $f: \mathbb{R}^n \rightarrow \mathbb{R}$ and $t \ge 0$, we let $P_t f(x) = \mathbf{E}_{\by \sim N_t(x)}[f(\by)]$.

\begin{definition}~\label{def:noise-sensitivity}
For any Boolean function $f: \mathbb{R}^n \rightarrow \{0,1\}$ and $t \ge 0$, we define
$\NS_t(f)$ (read as the ``noise sensitivity" of $f$) as
\[
\NS_t(f) = \Prx_{\bx \sim N(0,I_n), \by \sim N_t(\bx)} [f(\bx) \not  = f(\by)]. 
\]
For a set $A \subseteq \mathbb{R}^n$, we define $\NS_t(A)$ to be  $\NS_t(\mathbf{1}_A)$ where $\mathbf{1}_A(\cdot)$ is the indicator function of $A$. 
\end{definition}

We observe that using calls to $\SAMP(f)$ and $\MQ(f)$, we can efficiently  estimate the (normalized) noise sensitivity of $f$:
\begin{lemma} [Estimating the normalized noise sensitivity]~\label{lem:estimate-sensitivity}
For noise parameter $t>0$ and error parameter $\kappa>0$, there is an algorithm \EstSense~which for any 
 $f: \mathbb{R}^n \rightarrow \{0,1\}$, given query and sample access to $f$, 
 outputs a $\pm \kappa$
 accurate additive estimate to $\NS_t(f) / \Vol[f]$ with confidence $0.99$ using $\Theta(\kappa^{-2} )$ samples from $\SAMP(f)$ and the same number of queries to $\MQ(f)$. 
\end{lemma}

\begin{proof}
Sample $\bx \sim N(0,I_n)$ and $\by \sim N_t(\bx)$. Then, 
note that $$\NS_t(f) = \Pr \big[f(\bx) \not = f(\by)\big]=\Pr\big [ \bx \in f^{-1}(1) \land \by \not \in f^{-1}(1)\big] + \Pr \big[ \by \in f^{-1}(1) \land \bx \not \in f^{-1}(1)\big]. $$
As the distribution of $(\bx, \by)$ is exchangeable, it follows that 
\[
\NS_t(f) = 2 \Pr \big[ \bx \in f^{-1}(1) \land \by \not \in f^{-1}(1)\big] = 2 \Vol(f) \cdot \Pr\big[f(\by)= 0 \hspace{0.08cm}  | \hspace{0.08cm}  f(\bx)=1\big]. 
\]
Standard sampling shows that the RHS probabilty can be estimated to error $\pm \kappa/2$ with confidence $0.99$ using $\Theta(\kappa^{-2})$ calls to $\SAMP(f)$ and the same number of queries to $f$.   The lemma follows. 
\end{proof}

Another key ingredient we require is the following deep result of Barchiesi et~al.~\cite{barchiesi2017sharp} which obtains a sharp stability version of the Gaussian isoperimetric inequality. In particular, they establish the following: 
\begin{theorem}
[Main theorem of \cite{barchiesi2017sharp} and subsequent discussion: sharp stability version of~Gaussian isoperimetric inequality] \label{thm:BBJ}
Let $E$ be a measurable subset of $\mathbb{R}^n$ and let $s \in \R$ satisfy $\Vol(E) = \Phi(s)$, i.e., the Gaussian volume of $E$ is the same as that of the one-dimensional halfspace $\{x: x \leq s\}$. Define the \emph{Gaussian isoperimetric deficit} (namely, the gap between the two sides of the Gaussian isoperimetric inequality) to be
\[D(E) := \mathsf{surf}(E)-\varphi\big(\Phi^{-1}(\Vol(E)\big)=\mathsf{surf}(E) -\frac{1}{\sqrt{2\pi}}e^{-s^2/2}.\]
 Then, there is a halfspace $H = \{x : x \cdot \omega \leq  s\}$  (for some unit vector $\omega$) such that 
\[
\Vol\big(E \triangle H\big) ^2 \le C'(1+s^2) e^{-s^2/2}\cdot  D(E),
\]
 {where $C'>0$ is a universal constant (which can be taken to be at most $320\sqrt{2} \pi^{2}$).}
\end{theorem}
\begin{remark}
The main theorem in \cite{barchiesi2017sharp} defines the isoperimetric deficit $D(E)=\mathsf{surf}(E) - e^{-s^2/2}$ --- i.e. vis-a-vis their definition, we have an extra factor of $1/\sqrt{2\pi}$. The reason is just that while our definition of the surface area of a set $A$ is $\int_{x \in \partial A} \phi_n(x) d\sigma(x)$, the definition in \cite{barchiesi2017sharp} is equivalent to defining it as $\smash{\int_{x \in \partial A} \sqrt{2\pi} \phi_n(x) d\sigma(x)}$. 
\end{remark}

An easy corollary of this is the following: 
\begin{corollary}\label{corr:BBJ}
Let $E$ be a measurable subset of  $\mathbb{R}^n$ such that $\Vol(E)=p$ with $p \le 0.1$.
Suppose $$D(E) 
\le \frac{\beta^2 p}{{\log^{3/2}(1/p)}}$$ for some $\beta$. Then  there is a halfspace $H$ that satisfies $\Vol(H)=p$
and $\Vol(E \ \triangle \ H) \le C \beta p$
for some absolute constant $C>0$. 
(Recalling  $\varphi(\Phi^{-1}(p)) = \Theta(p \sqrt{\log(1/p)})$, an equivalent~statement~is~if $$D(E) \le \frac{\beta^2}{\log^2(1/p)}\cdot \varphi(\Phi^{-1}(p)),$$ 
then there is a halfspace $H$ that satisfies $\Vol(H)=p$ and  $\Vol(E \ \triangle \ H) \le C \beta p$.) 
\end{corollary}
\begin{proof}
Let $H = \{x: x_1 \le s\}$  be a one-dimensional halfspace such that $\Vol(H ) =p$. Note that $$\varphi (\Phi^{-1}(p)) = {\frac 1 {\sqrt{2 \pi}}}e^{-s^2/2}$$ 
so $e^{-s^2/2}=\Theta(\varphi(\Phi^{-1}(p)))$
  and using $\varphi(\Phi^{-1}(p)) = \Theta(p \sqrt{\log(1/p)})$,
   we have 
   $s^2=\Theta(\log (1/p))$. 
   
Now, applying \Cref{thm:BBJ}, we have that there is a halfspace
  $H$ with $\Vol(H)=p$ and
\begin{align*}
    \Vol(E \triangle H)^2 &\le O\left(\log (1/p)
    \cdot \varphi(\Phi^{-1}(p))\cdot \frac{\beta^2}{\log^2 (1/p)}
    \cdot \varphi(\Phi^{-1}(p))\right).
\end{align*}
This finishes the proof by plugging in 
$\varphi(\Phi^{-1}(p)) = \Theta(p \sqrt{\log(1/p)})$. 
\end{proof}

\subsection{Proof of \Cref{thm:Gaussian1}}

We now have all the necessary ingredients required to prove \Cref{thm:Gaussian1}. 
The algorithm alluded to in \Cref{thm:Gaussian1}, called \GSATest, is given in \Cref{alg:GSATest}.

\begin{algorithm}
\caption{A relative-error LTF tester over Gaussian space using samples and queries.}
\label{alg:GSATest}
\vspace{0.5em}
\textbf{Input:} $\SAMP(f)$, $\MQ(f)$, $\eps \in (0,1]$, $\hat{p} \in [(1+\zeta)^{-1}p, (1+\zeta)p]$ where $p = \vol(f) \leq 0.1$  
and $\zeta = c_1 \epsilon^2 /\log^2(1/p)$ for a suitably absolute constant $c_1 > 0$.  \\[0.5em]
\textbf{Output:} ``Accept'' or ``reject'' 

\ 

\GSATest$\pbra{\SAMP(f),\MQ(f), \hat{p}, \eps}$:
\begin{flushleft}\begin{enumerate}
    
    \item Set $p_{\mathsf{lb}} := \hat{p}/(1+ \zeta)$ (as in~\Cref{eq:plb}) and $\kappa,t$ (as in \Cref{eq:Definition-kappa,eq:def-t-xi}) as:
    \begin{align*}
\kappa := \frac{\zeta  \sqrt{t} \cdot \varphi(\Phi^{-1}(p_{\mathsf{lb}}))}{2\sqrt{\pi} \cdot p_{\mathsf{lb}}} \qquad\text{and}\qquad
t := \frac{c_2\epsilon^{10}}{\log^{10}(1/p_{\mathsf{lb}})}.
    \end{align*}
    \item Run the algorithm \EstSense~
    to compute a value $\balpha$, which, by \Cref{lem:estimate-sensitivity}, is a $\pm \kappa$-accurate estimate of $\NS_t(f)/\Vol[f]$ with confidence 0.99.
    
    \item If 
    \[
\balpha \leq \frac{2 \sqrt{t}}{\sqrt{\pi}} \cdot\frac{\varphi(\Phi^{-1}(p_{\mathsf{lb}}))}{p_{\mathsf{lb}}} +\kappa
\]
then output ``accept,'' otherwise output ``reject.''
\end{enumerate}
\end{flushleft}
\end{algorithm}

\begin{proofof}{\Cref{thm:Gaussian1}}
The algorithm is given a value $\hat{p}$ such that 
$p:=\Vol(f)$ satisfies
\begin{equation}~\label{eq:bound-p-hatp}
 \frac{p}{1 + \zeta} \le \hat{p} \le (1+ \zeta) p,
\end{equation}
where $\zeta = c_1 \epsilon^2 /\log^2(1/p)$ for a suitable absolute constant $c_1 > 0$.

The algorithm and its analysis use three main parameters, which are $\kappa$, $t$, and $\xi$. The precise values of $t$ and $\xi$ (in terms of $p$ and  $\epsilon$) are set later (see \eqref{eq:def-t-xi}), but looking ahead, importantly, both $\xi$ and $t$ are set to be $\poly(\log(1/p), 1/\epsilon)$.
The parameter $\kappa$ is defined as
\begin{equation}~\label{eq:Definition-kappa}
\kappa := \frac{\zeta  \sqrt{t}\cdot  \varphi(\Phi^{-1}(p_\mathsf{lb}))}{2\sqrt{\pi} \cdot p_{\mathsf{lb}}}
\le \frac{\zeta  \sqrt{t}\cdot  \varphi(\Phi^{-1}(p  ))}{ 2\sqrt{\pi} \cdot p }.
\end{equation}
using \Cref{fact:important-Gcdf} and $p_{\mathsf{lb}} \leq p$.
We also define the parameter $p_{\mathsf{lb}} $ as 
\begin{equation}
\label{eq:plb}
p_{\mathsf{lb}} :=  {\hat{p}}\big/({1+ \zeta}) .
\end{equation}

\medskip
\noindent {\bf Correctness analysis.} There are two parts of the analysis establishing correctness. We first start with the easy part, which is completeness. 

\begin{claim} [Completeness] \label{claim:broccoli-completeness}
If $f$ is a halfspace with $\Vol(f)=p$ 
and the input parameter $\hat{p}$ satisfies \eqref{eq:bound-p-hatp}, then $f$ passes the test with probability at least $0.99$.
\end{claim}
\begin{proof}
As $f$ is a halfspace, by \Cref{thm:noise-sensitivity-Ledoux} and \Cref{fact:halfspace-sa} we have
\[
\frac{\NS_t(f)}{\Vol(f)} \le \frac{2 \sqrt{t}}{\sqrt{\pi}} \cdot \frac{\varphi(\Phi^{-1}(p))}{p}.
\]
Note that $p_{\mathsf{lb}} \le p$ and the function $\Lambda (p) := \varphi(\Phi^{-1}(p))/p$ is decreasing in the interval $(0, 1/2)$.  Thus,  
\[
\frac{\NS_t(f)}{\Vol(f)} \le \frac{2 \sqrt{t}}{\sqrt{\pi}} \cdot \frac{\varphi(\Phi^{-1}(p_{\mathsf{lb}}))}{p_{\mathsf{lb}}}.
\]
As $\balpha \le \NS_t(f)/\Vol(f) + \kappa$ with probability $0.99$ by \Cref{lem:estimate-sensitivity}, this finishes the proof. 
\end{proof}

The next claim analyzes the soundness of the test. 
\begin{claim} [Soundness] \label{clm:soundness-surface-area-test}
Suppose $f : \mathbb{R}^n \rightarrow \{0,1\}$ has $\Vol(f)=p$, and suppose that $p$ and $\hat{p}$  satisfy \eqref{eq:bound-p-hatp}. If $\reldist(f,g) \ge \epsilon$ for every halfspace $g$, then $f$ fails the test with probability at least 0.99.
\end{claim}
\begin{proof}
We prove the contrapositive. Suppose $f$ passes the test with probability greater than $0.01$. By Line~2 of the algorithm and \Cref{lem:estimate-sensitivity}, it must be the case that
\[
\frac{\NS_t(f)}{\Vol(f)} \le \frac{2 \sqrt{t}}{\sqrt{\pi}} \cdot\frac{\varphi(\Phi^{-1}(p_{\mathsf{lb}}))}{p_{\mathsf{lb}}} +2 \kappa.
\]
Equivalently (recalling that $\Vol(f)=p$), we get that 
\[
\NS_t(f) \le \frac{2 \sqrt{t} p}{\sqrt{\pi}}\cdot  \frac{\varphi(\Phi^{-1}(p_{\mathsf{lb}}))}{p_{\mathsf{lb}}} +2 \kappa \cdot p. 
\]
Now, using \Cref{fact:important-Gcdf} and $p_{\mathsf{lb}} \leq p \le (1+\zeta)p_{\mathsf{lb}}$, we get that 
$$
\NS_t(f) \le\frac{2 \sqrt{t}p}{\sqrt{\pi}} \cdot\frac{\varphi(\Phi^{-1}(p))}{p} \cdot (1+\zeta) +2 \kappa \cdot p \nonumber = \frac{2 (1+\zeta) \sqrt{t}\cdot  \varphi(\Phi^{-1}(p))}{\sqrt{\pi}} + 2 \kappa \cdot p. 
$$ 
Plugging in the value of $\kappa$ from \eqref{eq:Definition-kappa} in the above estimate, we get 
\begin{equation}~\label{eq:sens1diff}
\NS_t(f) \le \frac{2  \sqrt{t}\cdot \varphi(\Phi^{-1}(p))}{\sqrt{\pi}} + \frac{3 \zeta  \sqrt{t}\cdot  \varphi(\Phi^{-1}(p))}{\sqrt{\pi}} \le \frac{3 \sqrt{t}\cdot  \varphi(\Phi^{-1}(p))}{\sqrt{\pi}},
 \end{equation}
 where we use the fact that $\zeta \le 1/3$ when $c_1$ is sufficiently small.
For any parameter $\xi>0$, we can now apply \Cref{thm:Neeman-reverse-Ledoux} to get that there is another function $\fsm:\mathbb{R}^n \rightarrow \zo$ satisfying the two bounds \Cref{eq:volume-diff} and \Cref{eq:fruit} given below:  {The first bound, using \Cref{eq:sens1diff}, is}
\begin{equation}~\label{eq:volume-diff}
 \Vol(f  \triangle \fsm) \le \frac{\NS_t(f)}{\xi} \le 
 \frac{3 \sqrt{t}\cdot \varphi(\Phi^{-1}(p))}{\sqrt{\pi} \cdot \xi}.
\end{equation}
For the second bound, let us now define $\tau$ as
\begin{equation} \label{eq:def-tau}\tau :=  \frac{3 \sqrt{t} \cdot \varphi(\Phi^{-1}(p))}{\sqrt{\pi} \cdot \xi \cdot p} =\Theta\pbra{\frac{\sqrt{t \log(1/p)}}{\xi}},
\end{equation}
where we use \Cref{fact:halfspace-sa}
in the last inequality. 
(Looking ahead to \Cref{eq:def-t-xi}, note that our choice of parameters will ensure that $\tau \le 1/10$.) By \Cref{eq:volume-diff}, we have that
\begin{equation} \label{eq:ppsmbound}
(1+ \tau) p \ge p_{\mathsf{sm}}:=\Vol(f_{\mathsf{sm}}) \ge (1-\tau) p.  
\end{equation}
The second bound given by \Cref{thm:Neeman-reverse-Ledoux} is
\begin{align}
    \mathsf{surf}(\fsm) &\le \sqrt{\frac{\pi}{2}} \bigg(1 + O\bigg(\frac{\xi}{\sqrt{\log(1/\xi)}} \bigg)\bigg) \cdot \frac{1}{\sqrt{e^{2t}-1}} \cdot \NS_t(f)  \label{eq:fruit}\\[1ex] 
    &\le  \sqrt{\frac{\pi}{2}} \cdot \frac{1}{\sqrt{e^{2t}-1}} \cdot \bigg( \frac{2  \sqrt{t} \cdot\varphi(\Phi^{-1}(p))}{\sqrt{\pi}} + \frac{3\zeta  \sqrt{t}\cdot  \varphi(\Phi^{-1}(p))}{\sqrt{\pi}} \bigg) +  \frac{\NS_t(f)}{\sqrt{e^{2t}-1}} \cdot  O\bigg(\frac{\xi}{\sqrt{\log(1/\xi)}} \bigg), \nonumber
\end{align}
where we used the first estimate of \eqref{eq:sens1diff}. Plugging in $e^{2t} -1 \ge 2t$
 and the second bound of \eqref{eq:sens1diff}, 
\begin{equation}~\label{eq:deficit-surf-area}
 \mathsf{surf}(\fsm)\le (1+3\zeta)\cdot \varphi(\Phi^{-1}(p)) + \frac{O(\xi) \cdot \varphi(\Phi^{-1}(p))}{\sqrt{\log(1/\xi)}}.
\end{equation}
Thus, defining the set $E:= f_{\mathsf{sm}}^{-1}(1)$, the isoperimetric deficit $D(E)$ for this set $E$ satisfies
\begin{eqnarray}
    D(E) = \mathsf{surf}(f_{\mathsf{sm}}) - \varphi(\Phi^{-1}(p_{\mathsf{sm}})) \le (1+3\zeta)\cdot \varphi(\Phi^{-1}(p)) + \frac{O(\xi) \cdot \varphi(\Phi^{-1}(p))}{\sqrt{\log(1/\xi)}} - \varphi(\Phi^{-1}(p_{\mathsf{sm}})), \nonumber
\end{eqnarray}
where we use \eqref{eq:deficit-surf-area} to get the inequality. Now, recall that by \Cref{eq:ppsmbound} we have $$p/p_{\mathsf{sm}} \le 1/(1-\tau) \le 1+2\tau$$ 
using $\tau \le 1/10$. 
Combining this with \Cref{fact:important-Gcdf}, it follows that $\varphi(\Phi^{-1}(p))\le (1+2\tau)\varphi(\Phi^{-1}(p_{\mathsf{sm}})$ and 
\begin{eqnarray}
    D(E) \le \left(2\tau+3\zeta(1+2\tau)+O\left(\frac{\xi}{\sqrt{\log(1/\xi)}}\right)\right)\cdot \varphi(\Phi^{-1}(p_{\mathsf{sm}})).
\end{eqnarray}
Now using \Cref{corr:BBJ} with this estimate (along with the fact $\Vol(E) = p_{\mathsf{sm}}$ and $\tau \le 1/10$), 
there is a halfspace $H$ such that 
\begin{equation} \label{eq:diff-H-E}
\Vol(H \, \triangle \,  E) = O(\beta p_{\mathsf{sm}}) = O(\beta p)
\end{equation}
where
\[
\beta = \log(1/p) \cdot O\bigg(\sqrt{\tau}  + \sqrt{\zeta} + \frac{\sqrt{\xi}}{\log^{1/4}(1/\xi)} \bigg). 
\]
Using \eqref{eq:def-tau},  it follows that
\begin{equation}~\label{eq:def-beta-bound}
\beta = \log(1/p) \cdot O\bigg(\frac{(t \log(1/p))^{1/4}}{\sqrt{\xi}}  + \sqrt{\zeta} + \frac{\sqrt{\xi}}{\log^{1/4}(1/\xi)} \bigg). 
\end{equation}
We can now use the fact that $\zeta = c_1 \epsilon^2 /\log^2(1/p)$ for a sufficiently small $c_1>0$ and finally set the parameters $t$ and $\xi$ as follows
\begin{equation}~\label{eq:def-t-xi}
t = \frac{c_2\epsilon^{10}}{\log^{10}(1/p_{\mathsf{lb}})} \qquad\text{and}\qquad  \xi = \frac{c_3\epsilon^{2.1}}{\log^{2.1}(1/p_{\mathsf{lb}})},
\end{equation}
for sufficiently small constants $c_2 \ll c_3 \ll 1$,
 to get from \Cref{eq:diff-H-E,eq:def-beta-bound} that
\begin{equation} \label{eq:almostthere}
\Vol (H \, \triangle \, E) \le {\frac \epsilon 2} \cdot p.
\end{equation}
Combining \Cref{eq:volume-diff} with \Cref{eq:almostthere}, and observing that by \Cref{eq:def-t-xi} we have that $\eqref{eq:volume-diff}$ is bounded from above by $ {  \eps  } p/2$ as well (with room to spare), by the triangle inequality we get that $\Vol(f^{-1}(1)  \triangle H) \leq \eps p$, i.e.~$f$ has relative distance at most $\eps$ from the halfspace $H$.
This establishes the contrapositive and concludes the proof of \Cref{clm:soundness-surface-area-test}.
\end{proof}

\noindent {\bf Complexity analysis.} 
We finally note that $\kappa$ defined in \eqref{eq:Definition-kappa} is $\Theta(\epsilon^{7}/ \log^{6.5}(1/p))$. As the 
sample complexity of~\Cref{lem:estimate-sensitivity} scales as $O(\kappa^{-2})$, our final sample complexity is $O(\eps^{-14} \cdot \log^{13}(1/p)).$
This concludes the proof of \Cref{thm:Gaussian1}.
\end{proofof}
 

%% file: sections/MORS-tester.tex

\section{Sample-based testing with known volume: Proof of \Cref{thm:Gaussian2}} \label{sec:rel-err-MORS}

In this section we prove \Cref{thm:Gaussian2}. 
Before we start, similar to \Cref{sec:broccoli}
we note that we can make a simplifying assumption, which is that $p :=\Vol(f)$ is at most $0.1$. This is because if $p>0.1$, we can just run the testing algorithm from \cite{Harms19} (\Cref{thm:mainHarms}) with error parameter $\epsilon / 10$; as $p>0.1$, standard-model $\eps/10$-testing implies $\eps$-relative-error testing. So for the rest of this section, we assume that $p \leq 0.1$. 

\subsection{Hermite analysis over $N(0, I_n)$}
\label{subsec:hermite}

Our notation and terminology follow Chapter~11 of ~\cite{odonnell-book}. We say that an $n$-dimensional \emph{multi-index} is a tuple $\alpha \in \N^n$, and we define 
\[
|\alpha| := \sum_{i=1}^n \alpha_i.
\]

For $n \in \N_{>0}$, we write $L^2(\R^n)$ to denote the space of functions $f: \R^n \to \R$ that have finite second moment under the Gaussian distribution, i.e. $f\in L^2(\R^n)$ if 
\[
\|f\|^2 = \Ex_{\bx \sim N(0,I_n)} \left[f(\bx)^2\right] < \infty.
\]
We view $L^2(\R^n)$ as an inner product space with 
\[\la f, g \ra := \Ex_{\bx \sim  N(0,I_n)}\big[f(\bx)g(\bx)\big].\]
We recall the Hermite basis for $L^2(\R)$:

\begin{definition}[Hermite basis]
	The \emph{Hermite polynomials} $(h_j)_{j\in\N}$ are the univariate polynomials 
	$$h_j(x) = \frac{(-1)^j}{\sqrt{j!}} \exp\left(\frac{x^2}{2}\right) \cdot \frac{d^j}{d x^j} \exp\left(-\frac{x^2}{2}\right).$$
\end{definition}

For example, we have
$$h_0(x)=1,\quad h_1(x) =x,\quad \text{and}\quad h_2(x)=\frac{x^2-1}{\sqrt{2}}.$$
The following fact is standard:

\begin{fact} [Proposition~11.33 of~\cite{odonnell-book}] \label{fact:hermite-orthonormality}
	The Hermite polynomials $(h_j)_{j\in\N}$ form a complete, orthonormal basis for $L^2(\R)$. 
 
 For $n > 1$, the collection of $n$-variate polynomials given by $(h_\alpha)_{\alpha\in\N^n}$ where
	$$h_\alpha(x) := \prod_{i=1}^n h_{\alpha_i}(x_i)$$
	forms a complete, orthonormal basis for $L^2(\R^n)$. 
\end{fact}

Given a function $f \in L^2(\R^n)$ and $\alpha \in \N^n$, we define its \emph{Hermite coefficient on} $\alpha$ as $\wh{f}(\alpha) = \la f, h_\alpha \ra$. It follows that $f:\R^n\to\R$ can be uniquely expressed as 
\[
    f = \sum_{\alpha\in\N^n} \wh{f}(\alpha)h_\alpha
\]
with the equality holding in $L^2(\R^n)$; we will refer to this expansion as the \emph{Hermite expansion} of $f$. One can check that Parseval's and Plancharel's identities hold in this setting:
\[\abra{f,f} = \sum_{\alpha\in \N^n}\wh{f}(\alpha)^2 \qquad\text{and}\qquad \abra{f,g} = \sum_{\alpha\in \N^n}\wh{f}(\alpha)\wh{g}(\alpha).\]
It is also readily verified that $\E[f(\bx)] = \wh{f}(0^n)$ and $\Var[f(\bx)] = \sum_{\alpha\neq 0^n} \wh{f}(\alpha)^2$ where ${\bx\sim N(0,I_n)}$.
 
Finally, we write $\smash{\bW^{=k}[f]}$ for the \emph{Hermite weight of $f$ at level-$k$}, i.e. 
\[
    \bW^{=k}[f] := \sum_{|\alpha| = k} \wh{f}(\alpha)^2,
\]
with $\bW^{\leq k}[f]$ defined similarly. 

We will require bounds on the level-$1$ and $2$ weights of Boolean-valued functions over $\R^n$; towards this, we recall the \emph{level-$k$ inequalities} (see e.g. Proposition~11 of~\cite{de2024detecting} or Section~9.5 of~\cite{odonnell-book} for more information): 

\begin{proposition}[Proposition~11 of~\cite{de2024detecting}] \label{prop:level-k}
    Suppose $f:\R^n\to\zo$ and $\smash{1 \leq k \leq 2\log\pbra{\frac{1}{\vol(f)}}}$. Then we have 
    \[
        \bW^{\leq k}[f] \leq \vol(f)^2\cdot\pbra{C \log\pbra{\frac{1}{\vol(f)}}}^k
    \]
    where $C>0$ is an absolute constant independent of $n$. 
\end{proposition}

We note that \Cref{prop:level-k} is an easy consequence of \emph{hypercontractivity} of the standard Gaussian random variables (cf. Chapters~9 and~11 of~\cite{odonnell-book}).

\subsection{Other preliminaries}
\label{subsec:mors-prelims}

We introduce the following notation:

\begin{definition}
\label{def:U}
Define the function $U: (0,1) \to [0, 1/2\pi]$ as
    \[
		U(p) := \pbra{\varphi(\Phi^{-1}(1-p))}^2.
    \]
\end{definition}

Note that $U(\cdot)$ is a function that is symmetric around $1/2$ with $U(1/2) = 1/2\pi$. 
It is easy to verify (using integration by parts) that $U(p)$ is the level-$1$ Hermite weight of an LTF with volume $p$ (or equivalently, of volume $1-p$):

\begin{fact}[Proposition~25 of~\cite{MORS10}]
\label{fact:meaning-of-U}
    Let $p \in (0,1)$ and let $t := \Phi^{-1}(p)$. Then 
    \[
        U(1-p) = \Ex_{\bx\sim N(0, 1)}\sbra{\mathbf{1}\cbra{\bx\geq t}\cdot \bx}^2 = \bW^{=1}[f]
    \]
    where $f:\R^n\to\R$ is the halfspace $f(x) = \mathbf{1}\{x\cdot v \geq t\}$ for any unit vector $v\in S^{n-1}$. 
\end{fact}

We will also need the following bound on $U$ from~\Cref{fact:halfspace-sa}: 

\begin{fact}[Proposition~24 of~\cite{MORS10}] \label{fact:U-asymptotics}
For any $p\in (0,1)$, $U(p) = \Theta(p^2 \ln(1/\min\{p, 1-p\}))$.
\end{fact}

\subsection{The main technical lemma}
\label{sec:our-mors-26}

The analysis of our sample-based testing algorithm, \HermiteTest, is described in \Cref{alg:MORS}. It   relies crucially on the following lemma:

\begin{lemma}
\label{lem:our-MORS}
	Suppose $f:\R^n\to\zo$ is a function with $p := \E[f]$ satisfying $0 < p < 0.1$, and let $t := \Phi^{-1}(1-p)$. For {$\ell \leq (10t)^{-1}$}, if 
	\[
    U(1-p)- \bW^{=1}[f]  \leq (\ell\,p\, t)^2
	\]
	holds, then $\reldist(f, g) \leq O(\ell\,t)$ for some LTF $g:\R^n\to\zo$. 
\end{lemma}

\Cref{lem:our-MORS} can be viewed as a relative-error strengthening of Theorem~26 of~\cite{MORS10}. Our proof is inspired by (and closely follows) that of Theorem~26 from~\cite{MORS10}; as will be clear from the proof, the coefficients of the LTF $g$ can be ``read off'' from degree-1 Hermite coefficients of $f$.

\begin{proof}
	Consider the function $h : \R^n \to \R$ 
 defined by 
	\[
		h(x) := \frac{1}{\sigma}\sumi \wh{f}(e_i)x_i - t \qquad\text{where}\qquad \sigma := \sqrt{\bW^{=1}[f]}.
	\]
	Note that 
        \begin{equation} \label{eq:vol-h}
            \Ex_{\bx\sim N(0, I_n)}\sbra{\mathbf{1}\cbra{h(\bx) \geq 0}} = \Ex_{\bx\sim N(0, I_n)}\sbra{\frac{1}{\sigma}\sumi \wh{f}(e_i)\bx_i \geq t} = \Ex_{\bx\sim N(0,1)}\sbra{\bx \geq t} = p,
        \end{equation}
        since $\sigma^2 = \sumi \wh{f}(e_i)^2$. 
        Let 
	\[
		\err := \Prx_{\bx\sim N(0,I_n)}\sbra{\Indicator\cbra{h(\bx) \geq 0} \neq f(\bx)}.
	\]
	We will show that $\err \leq O(p\,t\,\ell)$ which immediately implies the desired result. For this purpose,~by Parseval's formula, we have 
	\begin{equation} \label{eq:mors-26-fh-corr}
		\Ex_{\bx\sim N(0, I_n)}\big[f(\bx)\cdot h(\bx)\big] = \sigma - tp. 
	\end{equation}
	Furthermore, 
	\begin{align}
		\Ex_{\bx\sim N(0, I_n)}\sbra{h(\bx)\cdot\Indicator\cbra{h(\bx) \geq 0}}
		&= \Ex_{\bx\sim N(0,I_n)}\sbra{\pbra{\frac{1}{\sigma}\sumi \wh{f}(e_i)\bx_i}\Indicator\cbra{h(\bx) \geq 0}} -  tp \nonumber\\[0.8ex]		
		&= \Ex_{\bx\sim N(0,1)}\big[\bx\cdot \Indicator\cbra{\bx\geq t}\big]- tp\label{eq:mors-26-rot-inv-app} \\[0.4ex]
		&= \sqrt{U(1-p)} - tp,\label{eq:mors-26-h-ind-corr}
	\end{align}
	where~\Cref{eq:mors-26-rot-inv-app} uses the fact that a linear combination of independent Gaussian random variables is itself a Gaussian (with variance equal to the sum of the squared weights, in this case~$1$)  and~\Cref{eq:mors-26-rot-inv-app} relies on~\Cref{fact:meaning-of-U}.
	
    Note that 
	\begin{equation} \label{eq:spider}
		h(g)\pbra{\Indicator\cbra{h(g) \geq 0} - f(g)}
		= 
		\begin{cases}
			|h(g)| & \Indicator\cbra{h(g) \geq 0} \neq f(g) \\ 
			0 & \Indicator\cbra{h(g) \geq 0} = f(g)
		\end{cases}.
	\end{equation}
	Combining~\Cref{eq:mors-26-fh-corr,eq:mors-26-h-ind-corr}, we get 
	\begin{align}
		\Ex_{\bx\sim N(0,I_n)}\sbra{h(\bx)\pbra{\Indicator\cbra{h(\bx) \geq 0} - f(\bx)}} 
		&= \sqrt{U(1-p)} - \sigma \nonumber \\[-1.5ex]
		&\leq \frac{(p\,t\,\ell)^2}{\sqrt{U(1-p)}} \label{eq:mors-26-fact-8-app} \\[0.3ex]
		&\leq O(p\,t\,\ell^2) \label{eq:mors-26-gaussian-tails-app}
	\end{align}
	where~\Cref{eq:mors-26-fact-8-app} used the following simple \Cref{fact:simple} and \Cref{eq:mors-26-gaussian-tails-app} used \Cref{prop:gaussian-tails}:
 \[
    \sqrt{U(1-p)} 
    = \phi\circ\Phi^{-1}(p) 
    = \phi(t)
    \geq t\cdot \Phi(-t) 
    = t p.
 \]
\begin{fact} \label{fact:simple}
Let $a>0$ and $b,\eps\ge 0$ such that $a-b\le  \eps $. Then we have $\sqrt{a}-\sqrt{b}\le \eps\big/\sqrt{a}$.
\end{fact}
\begin{proof}
The claim is trivial if $b>a$ since the LHS is negative, so 
  we assume without loss of generality that $b\le a$.
Then we have $\eps\ge a-b\ge a-\sqrt{ab}=\sqrt{a}(\sqrt{a}-\sqrt{b})$ and the inequality follows.
\end{proof}
 
We will next show that 
\begin{equation} \label{eq:mors-26-final-goal}
    \Prx_{\bx\sim N(0,I_n)}\sbra{|h(\bx)| \leq \frac{\ell}{2}} = O(p\,t\,\ell), 
\end{equation}
which completes the proof. To see this, it follows from~\Cref{eq:spider,eq:mors-26-gaussian-tails-app,eq:mors-26-final-goal} that
\[
    \frac{\ell}{2}\cdot\pbra{\err - \Prx_{\bx\sim N(0,I_n)}\sbra{|h(\bx)| \leq \frac{\ell}{2}} } \leq O(p\,t\,\ell^2) \qquad\text{and so}\qquad \err = O(p\,t\,\ell)
\] 
as desired. 

We will now establish~\Cref{eq:mors-26-final-goal}. Note that  
\begin{align}
    \Prx_{\bx\sim N(0,I_n)}\sbra{|h(\bx)| \leq \frac{\ell}{2}} 
    &= \Prx_{\bx\sim N(0,I_n)}\sbra{\frac{1}{\sigma}\sumi \wh{f}(e_i)\bx_i \in \sbra{t \pm \frac{\ell}{2}}} \nonumber \\[1ex]
    &= \Prx_{\bx\sim N(0, 1)}\sbra{\bx \in \sbra{t \pm \frac{\ell}{2}}} \nonumber \\[-0.5ex]
    &\leq {\ell}\cdot\phi\pbra{t - \frac{\ell}{2}} = \ell\cdot\phi(t)\cdot\frac{\phi\pbra{t - \frac{\ell}{2}}}{\phi(t)}. \label{eq:mors-26-orange}
\end{align}
\Cref{eq:mors-26-orange} is thanks to {$t \geq \ell/2$} (by assumption on $t$ and $\ell$)
and the fact that $\phi(\cdot)$ is decreasing on $(0, \infty)$. 
Note that 
\[
    \frac{\phi\pbra{t - \frac{\ell}{2}}}{\phi(t)} = {\exp\pbra{- \frac{\ell^2}{8} + \frac{t\ell}{2}} = O(1)}
\]
thanks to the assumption that {$t\ell \leq 0.1$}. Plugging this back into~\Cref{eq:mors-26-orange} and using 
\[
    \phi(t)\pbra{\frac{1}{t} - \frac{1}{t^3}} \leq 1-\Phi(t) = p,
\]
which is a consequence of~\Cref{prop:gaussian-tails}, we get 
\[
    \Prx_{\bx\sim N(0,I_n)}\sbra{|h(\bx)| \leq \frac{\ell}{2}} \leq O(\ell\cdot\phi(t)) = O\pbra{\frac{\ell\, p\, t}{1- ({1}/{t^2})}} = O(\ell\,p\,t),
\]
where we relied on the fact that $p \leq {0.1}$ (and so $t \geq 1.01$), which establishes~\Cref{eq:mors-26-final-goal}.
\end{proof}

\subsection{Proof of \Cref{thm:Gaussian2}}
\label{subsec:proof-our-MORS}

We record the following easy corollary of \Cref{fact:important-Gcdf}: 

\begin{corollary}  \label{cor:U-lipschitzness}
    Suppose that $\hat{p} \in [(1+\tau)^{-1}p, (1+\tau)p]$. Then we have
    \[
        \frac{U(\wh{p})}{\hat{p}^2} \in \sbra{\frac{1}{(1+\tau)^4}\frac{U(p)}{p^2},  {(1+\tau)^4}\frac{U(p)}{p^2}}.
    \]
\end{corollary}

\begin{proof}
    As an immediate consequence of~\Cref{fact:important-Gcdf}, we have 
    \[
        U(\hat{p}) \in \sbra{(1+\tau)^{-2}U(p), (1+\tau)^2U(p)}.
    \]
    The result now follows thanks to $\hat{p} \in [(1+\tau)^{-1}p, (1+\tau)p]$. 
\end{proof}

We can now turn to the proof of~\Cref{thm:Gaussian2}:

\begin{algorithm}
\caption{The sample-based relative-error LTF tester over Gaussian space.}
\label{alg:MORS}
\vspace{0.5em}
\textbf{Input:} $\SAMP(f)$, $\hat{p} \in [(1+\eta)^{-1}p, (1+\eta)p]$ where $p = \vol(f) \leq 0.1$, $\eps \in (0,1]$ and $\eta = c_2\eps^2/\log(1/p)$ for a suitably small absolute constant $c_2>0$.\\[0.5em]
\textbf{Output:} ``Accept'' or ``reject'' 

\ 

\HermiteTest$\pbra{\SAMP(f), \hat{p}, \eps}$:
\begin{enumerate}
    \item Set 
    \[
        m :={\frac {\sqrt{n}}{\eps^2}} + {\frac {\log^2(1/p)}{\eps^4}},
    \]
    and draw $\x{1}, \ldots, \x{m}, \y{1}, \ldots, \y{m} \leftarrow \SAMP(f)$. 
    
    \item Compute 
    	\[
    		\bT := \frac{1}{m^2} \sum_{i, j = 1}^m \x{i}\cdot \y{j}. 
    	\]
    \item Let 
    \[
        \tau := \frac{U(\wh{p})}{\wh{p}^2},
    \]
    and output ``accept'' if $|\bT - \tau| \leq c\eps^2$ for a constant $c$ implicit in the proof of\newline \Cref{thm:Gaussian2}; output ``reject'' otherwise. 
\end{enumerate}

\end{algorithm}

\begin{proof}[Proof of~\Cref{thm:Gaussian2}]
Thanks to the discussion at the start of~\Cref{sec:rel-err-MORS}, we may assume without loss of generality that $p = \Vol(f) < 0.1$. 

We first characterize the expectation of the estimator $\bT$. To do this, note that 
\begin{align}
	\Ex[\bT] 
	&= \frac{1}{m^2}\Ex\sbra{\sum_{i, j=1}^m \x{i}\cdot\y{j}} \nonumber \\
	&= \sum_{\ell = 1}^n \Ex_{\bx, \by \sim \SAMP(f)}\sbra{\bx_\ell\cdot\by_\ell} \nonumber \\ 
	&= \sum_{\ell = 1}^n \frac{1}{p^2}\Ex_{\bx,\by\sim N(0, I_n)}\sbra{f(\bx)\bx_\ell \cdot f(\by) \by_\ell} \nonumber \\
	&= \sum_{\ell = 1}^n \pbra{\frac{1}{p}\cdot\Ex_{\bx\sim N(0, I_n)}\sbra{f(\bx)\cdot\bx_\ell}}^2 \label{eq:hermite-calc-indep} \\[0.6ex]
	&= \frac{\bW^{=1}[f]}{p^2}, \label{eq:dlns-mean-formula}
\end{align}
where \Cref{eq:hermite-calc-indep} relies on the independence of the samples from $\SAMP(f)$. 

Next, we will show that for any $f:\R^n\to\zo$ with $\vol(f) = p$, we have 
\begin{equation} \label{eq:var-calc-goal}
    \Varx\sbra{\bT} = O\pbra{\frac{1}{m}\log^2\pbra{\frac{1}{p}} + \frac{1}{m^2}\pbra{n + \log^2\pbra{\frac{1}{p}}}} = O\pbra{\frac{\log^2(1/p)}{m} + \frac{n}{m^2} }.
\end{equation}
In particular, for $m$ as in the statement of~\Cref{thm:Gaussian2}, we have
\begin{equation}
    \Var[\bT] \leq O(\eps^4). \label{eq:actual-variance-bound-eps-fourth}
\end{equation}

We first show how \Cref{eq:var-calc-goal} implies \Cref{thm:Gaussian2} before turning to its proof. 
Note that: 
\begin{itemize}
    \item If $f$ is an LTF,  then using \Cref{eq:dlns-mean-formula}, \Cref{cor:U-lipschitzness}, \Cref{fact:meaning-of-U} and~\Cref{fact:U-asymptotics}, we have 
    \[
        \abs{\tau - \E [{\bT}]} 
        = \abs{\frac{U(\hat{p})}{\hat{p}^2} -\frac{U(p)}{p^2}} 
        \leq  O(\eta)\cdot\frac{U(p)}{p^2} = O(\eta)\cdot\log\pbra{\frac{1}{p}}.
    \]

\item If $\reldist(f, g) \geq \eps$ for every LTF $g:\R^n\to\zo$, then it follows from the triangle inequality, \Cref{eq:dlns-mean-formula}, the contrapositive of \Cref{lem:our-MORS}, \Cref{cor:U-lipschitzness}, \Cref{fact:meaning-of-U} and~\Cref{fact:U-asymptotics} that 
\[
    \abs{\tau- \E[\bT]} \geq \abs{ \frac{U(p)}{p^2} - \E[\bT]} - \abs{\frac{U(\hat{p})}{\hat{p}^2} -\frac{U(p)}{p^2}} = \Omega(\eps^2) - O(\eta)\cdot\log\pbra{\frac{1}{p}},
\]
where we set $\ell$ in \Cref{lem:our-MORS} to be 
$\eps/(10t)$ so that $\ell\le (10t)^{-1}$ holds.
\end{itemize}
The theorem follows by Chebyshev's inequality thanks to~\Cref{eq:actual-variance-bound-eps-fourth} and  $\eta = c_2\eps^2/\log(1/p)$.

The remainder of the proof will establish~\Cref{eq:var-calc-goal}.
Our proof will follow a similar strategy to that of Theorem~13 of~\cite{de2024detecting}. 
First, note that
\[
    \Varx[\bT] = \frac{1}{m^4}\cdot\sum_{i,j,k,\ell = 1}^T \Cov\pbra{\x{i}\cdot\y{j}, \x{k}\cdot\y{\ell}}
\]
where $\Cov(\bX, \bY) = \E[\bX\cdot\bY] - \E[\bX]\cdot\E[\bY]$ is the covariance of the random variables $\bX$ and $\bY$. 
In particular, we have 
\begin{align}
    \Cov\pbra{\x{i}\cdot\y{j}, \x{k}\cdot\y{\ell}}
    &= \Ex\sbra{(\x{i}\cdot\y{j})\cdot (\x{k}\cdot\y{\ell})} - \Ex\sbra{(\x{i}\cdot\y{j})}^2         \nonumber \\
    &\leq \Ex\sbra{(\x{i}\cdot\y{j})\cdot (\x{k}\cdot\y{\ell})}      \label{eq:drop-covariance-second-term}
\end{align}
where all the random variables are drawn from  $\SAMP(f)$.
Note that
\begin{flushleft}
\begin{itemize}
    \item If $i \neq k$ and $j \neq \ell$, then $\Cov\pbra{\x{i}\cdot\y{j}, \x{k}\cdot\y{\ell}} = 0$ thanks to independence.

    \item If $i = k$ but $j \neq \ell$, then 
    \begin{align}
        \Ex\sbra{(\x{i}\cdot\y{j})\cdot (\x{k}\cdot\y{\ell})}
        &= \Ex\sbra{\y{j}\cdot\underbrace{\Ex\sbra{\x{i}\cdot(\x{i})^\top}}_{=:A}\cdot \y{\ell}} \nonumber \\ 
        &= \Ex\sbra{\y{j}} A\Ex\sbra{\y{\ell}} \nonumber \\
        &\leq \vabs{\Ex\sbra{\y{\ell}}}_2^2 \cdot \|A\|_{\mathrm{op}} \nonumber \\
        &= \frac{\bW^{=1}[f]}{p^2}\cdot \|A\|_{\mathrm{op}}\label{eq:lego}
    \end{align}
    where~\Cref{eq:lego} relies on~\Cref{eq:dlns-mean-formula}. 
    Note that $A=(A_{i,j})$ is an $n\times n$ matrix with 
    \begin{equation} \label{eq:A-def}
        A_{i,j} = \Ex_{\bx\sim\SAMP(f)}\sbra{\bx_i\cdot\bx_j}.
    \end{equation}
    We will control $\|A\|_{\mathrm{op}}$ shortly; for now, note that the case when $i \neq k$ but $j = \ell$ is identical by symmetry.

    \item If $i = k$ and $j = \ell$, then 
    \begin{align}
         \Ex\sbra{(\x{i}\cdot\y{j})\cdot (\x{k}\cdot\y{\ell})} 
         &= \Ex\sbra{(\bx\cdot\by)^2} \nonumber \\
         &= \Ex\sbra{\sum_{i,j=1}^n \bx_i\bx_j \by_i\by_j} \nonumber \\
         &= \sum_{i,j=1}^n \Ex\sbra{\bx_i\bx_j}^2 = \|A\|_F^2. \label{eq:case-3-frob}
    \end{align}
\end{itemize}\end{flushleft}
It follows from \Cref{eq:lego,eq:case-3-frob} that 
\begin{equation} \label{eq:var-almost-there}
    \Var[\bT] \leq \frac{1}{m^4}\pbra{m^3\pbra{\frac{\bW^{=1}[f]\cdot\|A\|_{\mathrm{op}}}{p^2}} + m^2\|A\|_F^2}.
\end{equation}

We will rely on the level-$k$ inequality (\Cref{prop:level-k}) to control both $\|A\|_{\mathrm{op}}$ as well as $\|A\|_F$, starting with the former.
First, note that 
\[
    \|A\|_{\mathrm{op}} \leq \vabs{\mathrm{diag}(A)}_{\mathrm{op}} + \vabs{A - \mathrm{diag}(A)}_{\mathrm{op}} \leq\vabs{\mathrm{diag}(A)}_{\mathrm{op}} + \vabs{A - \mathrm{diag}(A)}_F
\]
where $\mathrm{diag}(A)$ is the $n\times n$ diagonal matrix given by the diagonal entries of $A$.  

We will first give an upper bound on $\vabs{\mathrm{diag}(A)}_{\mathrm{op}}$:
\begin{align*}
    {\mathrm{diag}(A)}_{i,i} = \Ex_{\bx\sim\SAMP(f)}\sbra{\bx_i^2} 
    &= \frac{1}{p} \Ex_{\bx\sim N(0, I_n)}\sbra{f(\bx)\cdot\bx_i^2}.
\end{align*}    
Given that $x^2=\sqrt{2}h_2(x)+h_0(x),$
  we have by Cauchy-Schwarz that
$$
\mathrm{diag}(A)_{i,i}=\frac{1}{p}\cdot \left(\hat{f}(0^n)+\sqrt{2}\cdot \hat{f}(2e_1)\right)\le O(1)\cdot \frac{1}{p}\cdot \sqrt{\bW^{\leq 2}[f]}\le O\left(\log\left(\frac{1}{p}\right)\right),
$$
where the last inequality used \Cref{prop:level-k}.
It immediately follows that
\begin{equation} \label{eq:diag-op-norm}
    \vabs{\mathrm{diag}(A)}_{\mathrm{op}} \leq O\pbra{\log\pbra{\frac{1}{p}}}. 
\end{equation}

Turning to $\|A - \mathrm{diag}(A)\|_F$, we have
\begin{align}
    \vabs{A - \mathrm{diag}(A)}_F^2 
    &= \sum_{i\neq j} \Ex_{\bx\sim\SAMP(f)}\sbra{\bx_i\cdot\bx_j}^2 \nonumber \\
    &= \frac{1}{p^2} \sum_{i\neq j}\Ex_{\bx\sim N(0, I_n)}\sbra{f(\bx)\cdot\bx_i\bx_j}^2 \nonumber \\
    &= \frac{1}{p^2}\sum_{i\neq j} \wh{f}(e_i + e_j)^2 \nonumber \\
    &\leq O\pbra{\log^2\pbra{\frac{1}{p}}} \label{eq:off-diag-frob},
\end{align}
where~\Cref{eq:off-diag-frob} relied on the level-$k$ inequality (\Cref{prop:level-k}) again. 

Combining~\Cref{eq:diag-op-norm,eq:off-diag-frob} and applying \Cref{prop:level-k} on $\bW^{=1}[f]$ yields
\begin{equation} \label{eq:case2-var}
    \frac{\bW^{=1}[f]\cdot\|A\|_{\mathrm{op}}}{p^2} 
    \leq 
    \frac{\bW^{=1}[f]}{p^2}\cdot O\pbra{\log\pbra{\frac{1}{p}}}
    \leq 
    O\pbra{\log^2\pbra{\frac{1}{p}}}.
\end{equation}

Finally, thanks to~\Cref{eq:off-diag-frob}, we can bound $\|A\|_F^2$ as follows:
\begin{align*}
    \|A\|_F^2 &\leq O\pbra{\log^2\pbra{\frac{1}{p}}} + \sumi \Ex_{\bx\sim\SAMP(f)}\sbra{\bx_i^2}^2 
    =O\pbra{\log^2\pbra{\frac{1}{p}}}+\frac{1}{p^2}
    \sum_{i=1}^n\left(\hat{f}(0^n)+\sqrt{2}\hat{f} (2e_i)\right)^2.
\end{align*}
Using Cauchy-Schwarz, we have
$$
\sum_{i=1}^n\left(\hat{f}(0^n)+\sqrt{2}\hat{f}(2e_i)\right)^2
\le O(1)\sum_{i=1}^n \hat{f}(0^n)^2+\hat{f}(2e_i)^2
\le O\left(np^2+\bW^{=2}[f]\right).
$$
Plugging in \Cref{prop:level-k} and
combining these two inequalities, we have
\begin{equation}
    \|A\|_F^2 \leq  
    O\pbra{n + \log^2\pbra{\frac{1}{p}}}. 
    \label{eq:case3-var}
\end{equation}

\Cref{eq:var-calc-goal} now follows immediately from~\Cref{eq:var-almost-there,eq:case2-var,eq:case3-var}, which completes~the proof of~\Cref{thm:Gaussian2}.
\end{proof}

%% file: sections/combined-tester.tex

\section{Testing with unknown volume: Proof of \Cref{thm:Gaussian3}}
\label{subsec:unknown-p-testing}

In this section we prove  \Cref{thm:Gaussian3}. 
 Throughout this section we write $p$ to denote $\Vol(f)$, and we assume that $p \leq 0.01$.  This is without loss of generality, because by making $O(1)$ queries on random points from $N(0,I_n)$ we can get an estimate of $p$ that is additively accurate to $\pm 0.001$ with high probability, and if this estimate is larger than (say) $0.005$ then we can run the $O(1/\eps^5)$-query tester of \cite{MosselNeeman15}. 

\subsection{Structural ingredients} \label{combined-structural}

We now list some of the fundamental structural theorems we will need in this section. 
The first are the results of Borell and Sudakov-Tsirelson~\cite{Borell:75, Borell:85, ST:78} which show that for both noise sensitivity and surface area, halfspaces are the minimizers. 
\begin{theorem}[Halfspaces are minimizers of noise sensitivity and surface area]\label{thm:Borell}
Let $A \subseteq \mathbb{R}^n$ be a measurable set with  $\Vol(A) = p$ and $t> 0$. Let $H_p$ be a halfspace be such that $\Vol(H_p)=p$. Then, 
\begin{enumerate}
    \item $\mathsf{surf}(A) \ge \mathsf{surf}(H_p)$. 
    \item $\NS_t(A) \ge \NS_t(H_p)$.
\end{enumerate}
\end{theorem}

We will use the above theorem  and \Cref{thm:Neeman-reverse-Ledoux} to get a fairly sharp lower bound on the noise sensitivity~of a halfspace of volume $p$ over the Gaussian space. 
\begin{corollary}\label{corr:bounds-on-stability-for-halfspaces}
    Let $f: \mathbb{R}^n \rightarrow \{0,1\}$ be a halfspace with $\Vol(f)= p\le 0.1$, and let $t:0 \le t <1/2$ be such that $t \log (1/p) \le c$ for some sufficiently small constant $c>0$. Let $\mu := t \log(1/p)$. Then 
    \[
   \frac{2\sqrt{t}}{\sqrt{\pi}} \varphi(\Phi^{-1}(p)) \cdot \big(1-\Theta(\mu^{1/4})\big) \le \NS_t(f) \le \frac{2\sqrt{t}}{\sqrt{\pi}} \varphi(\Phi^{-1}(p)). 
    \]
\end{corollary}
\begin{proof}
The upper bound is a consequence of \Cref{thm:noise-sensitivity-Ledoux} and \Cref{fact:halfspace-sa}. To get the lower bound, note that for any parameter
$\beta >0$, by \Cref{thm:Neeman-reverse-Ledoux}, there exists a set $B$ satisfying item (2) of \Cref{thm:Neeman-reverse-Ledoux} such that 
\[
\Vol(B) \ge p - \frac{\NS_t(A)}{\beta} \ge p - \frac{2 \sqrt{t}}{\sqrt{\pi} \beta} \cdot \varphi(\Phi^{-1}(p)). 
\]
If we set $\beta \asymp \frac{\sqrt{t \log(1/p)}}{\delta}$, then using \Cref{fact:halfspace-sa} we get that 
\[
\frac{2 \sqrt{t}}{\sqrt{\pi} \beta} \cdot \varphi(\Phi^{-1}(p)) \le \delta p,
\quad \text{and hence} \quad  
\Vol(B)  \ge (1-\delta) p. 
\]
Note that by Borell's isoperimetry theorem (Item~1 of \Cref{thm:Borell}), if $\Vol(B)=q \geq (1-\delta)p$, then 
\begin{equation}~\label{eq:lb-surfacearea}
\mathsf{surf} (B) \ge \varphi (\Phi^{-1}(q)) \ge (1-\delta)\varphi(\Phi^{-1}(p)),
\end{equation}
where we used \Cref{fact:important-Gcdf} for the second inequality. Now recalling that $B$ satisfies item (2) of \Cref{thm:Neeman-reverse-Ledoux}, we get that
\[
\NS_t(A) \ge \frac{\sqrt{2(e^{2t}-1)}}{\sqrt{\pi}} \cdot \mathsf{surf}(B) \cdot \frac{1}{ \left(1 + O\left(\frac{\beta}{\sqrt{\log(1/\beta)}}\right)\right)}.
\]
Plugging \eqref{eq:lb-surfacearea} into the above and using $e^{2t} - 1 \geq 2t,$ we get 
\[
\NS_t(A) \geq \frac{2\sqrt{t}}{\sqrt{\pi}} \varphi(\Phi^{-1}(p)) 
\cdot (1-\delta) \cdot (1-O(\beta)) 
\ge \frac{2\sqrt{t}}{\sqrt{\pi}} \varphi(\Phi^{-1}(p)) 
\cdot  (1-O(\beta  + \delta)).
\]
 Since $\beta \asymp \frac{\sqrt{t  \log(1/p)}}{\delta}$, we can set $\beta = \Theta((t \log (1/p)^{1/4})$ to obtain $\delta=\Theta(\beta)$ and
\[
\NS_t(A) \geq \frac{2\sqrt{t}}{\sqrt{\pi}} \varphi(\Phi^{-1}(p))  \cdot \left(1-\Theta\left((t \log (1/p))^{1/4} \right)\right).
\]
This finishes the proof. 
\end{proof}

Next, we prove \Cref{lem:our-MORS2}
  which is essentially the contrapositive of 
  \Cref{lem:our-MORS}:

\begin{lemma}
\label{lem:our-MORS2}
    Let $f:\R^n\to\zo$ be a function with $p = \Vol(f)$ satisfying $0 < p < 0.1$. Assume~that $\reldist(f,g) \ge \eps $ for every LTF $g$ (where $\eps$ is at most some suitably small absolute constant). Then
	\[
		\sqrt{U(p)} - \sqrt{\bW^{=1}[f]} \ge \Omega\left({\frac {(\eps p)^2} { \sqrt{U(p)}}}\right).
	\]
\end{lemma}
\begin{proof}
    The contrapositive formulation of \Cref{lem:our-MORS} gives us (recall $U(p)=U(1-p)$) that
\[
 U(p) - \bW^{=1}[f] \ge \Omega((\eps p)^2),
\]
which is equivalent to
\[
\sqrt{U(p)} - \sqrt{\bW^{=1}[f]} \ge\Omega\left( {\frac {(\eps p)^2} { \sqrt{U(p)} + \sqrt{\bW^{=1}[f]}}}\right)
\geq
\Omega\left({\frac {(\eps p)^2} { \sqrt{U(p)}}}\right),
\]
where for the last inequality we again used that $U(p) \geq \bW^{=1}[f].$
\end{proof}

We define the function $V: (0,1) \to \R_{\geq 0}$ as follows:
\begin{equation} \label{eq:def-of-V}
V(p) := {\frac {U(p)}{p^2}}, \quad\text{so} \quad V(p)=\Theta(\ln(1/p))
\text{ for $0 < p \leq 1/2$ by \Cref{fact:U-asymptotics}}.
\end{equation}
Looking ahead, our algorithm will form a ``guess'' $p_2$ of the true value of $p := \Vol(f)$ by inverting the function $V$ on an estimate of ${  {U(p)}/{p^2}}$ that is obtained from samples. To analyze this algorithm we need a structural result which gives us control on how much $V^{-1}$ can deviate from the true value of $V^{-1}({ {U(p)}/{p^2}})=p$ when it is evaluated on an estimate of ${ {U(p)}/{p^2}}$ rather than the exact value; this result is given in \Cref{claim:v-inverse}.
(We remark that for a quick intuitive understanding of \Cref{lemma:stitch-bostanci}, \Cref{clm:bound-psi}, \Cref{lem:bounds-on-V} and \Cref{claim:v-inverse} the reader may find it useful to plug in the asymptotic bounds provided in \Cref{eq:def-of-V,eq:def-of-f}, but we also give full proofs below.)

We require some setup before embarking on the proof of \Cref{claim:v-inverse}. 
For convenience, define the function $\psi: (0,1)\to\R$ as 
\begin{equation} \label{eq:def-of-f}
    \psi(p) = \frac{I(p)}{p},
    \quad\text{so} \quad \psi(p)=\Theta(\sqrt{\ln(1/p)})
\text{ for $0 < p \leq 1/2$ by \Cref{fact:U-asymptotics}},
\end{equation}
where $I(p) = \phi(\Phi^{-1}(p))$ is the Gaussian isoperimetric function (\Cref{eq:I}) with $V(p) = \psi(p)^2$. 

\begin{fact} \label{fact:moo-deng}
    We have $I'(p) = -\Phi^{-1}(p)$.
\end{fact}

\begin{proof}
    This is a calculus exercise: since $I(p)=\phi(\Phi^{-1}(p)),$ taking the derivative we get
    \[
    I'(p) = \phi'(\Phi^{-1}(p)) \cdot (\Phi^{-1})'(p)
    =-\Phi^{-1}(p) \cdot \phi(\Phi^{-1}(p)) \cdot {\frac 1 {\phi(\Phi^{-1}(p))}} = -\Phi^{-1}(p),
    \]
   where the second equality used $\phi(x) = {\frac 1 {\sqrt{2 \pi}}} e^{-x^2/2}$ and the fact that $$(\Phi^{-1})'(p) = {\frac 1 {\Phi'(\Phi^{-1}(p))}} = {\frac 1 {\phi(\Phi^{-1}(p)}}.$$
   This finishes the proof of the statement.
\end{proof}

We will use the following estimate for $\psi'(p)$:

\begin{lemma} \label{lemma:stitch-bostanci}
    Suppose that $p \in (0, 0.1]$. We have 
    \[
        \psi'(p) = -\Theta\pbra{\frac{1}{p\sqrt{\log\pbra{ {1}/{p}}}}}.
    \]
\end{lemma}

\begin{proof}
    For convenience, let $-z := \Phi^{-1}(p)$. Note that standard Gaussian tail bounds (\Cref{prop:gaussian-tails}) imply that 
    \[
        z = \Theta\pbra{\sqrt{\log\pbra{1/p}}}.
    \]
    Using the chain rule (recall the definition of $\psi$, cf.~\Cref{eq:def-of-f}) and~\Cref{fact:moo-deng}, we get 
    \begin{align}
        \psi'(p) 
         = \frac{-1}{p}\pbra{\frac{\phi(z)}{p} - z}  
         = -\frac{z}{p}\left(\frac{\phi(z)}{pz}-1\right)
         \label{eq:happy-halloween}
    \end{align}
    Next, note that~\Cref{prop:gaussian-tails} implies  
    \[
        \phi(z)\pbra{\frac{1}{z} - \frac{1}{z^3}} 
        \leq p 
        \leq 
        \phi(z)\pbra{\frac{1}{z} - \frac{1}{z^3} + \frac{3}{z^5}},~~\text{i.e.}~~
          {1 - \frac{1}{z^2}}
        \leq 
        \frac{zp}{\phi(z)} 
        \leq 
          {1 - \frac{1}{z^2} + \frac{3}{z^4}}.
    \]
    Plugging this into~\Cref{eq:happy-halloween} then gives
    \[
    -\frac{z}{p}\pbra{\frac{1}{z^2-1}} 
    \geq 
    \psi'(p)
    \geq 
    -\frac{z}{p}\pbra{\frac{z^2 - 3}{z^4 - z^2 + 3}},
    \]
    and so $\psi'(p) = -\Theta(1/zp)$. 
    Recalling that $z = \Theta(\sqrt{\log(1/p)})$ completes the proof. 
\end{proof}

We also require the following bounds on the function $\psi(\cdot)$. 
\begin{claim}~\label{clm:bound-psi}
In the interval $(0,0.1)$, the function $\psi(\cdot)$ is monotonically decreasing. Furthermore, if $p_2 > p_1$, then, 
\[
    \psi(p_1) - \psi(p_2) = \Theta \bigg( \sqrt{\log(1/p_1)} - \sqrt{\log(1/p_2)} \bigg). 
\]
\end{claim}
\begin{proof}
\Cref{lemma:stitch-bostanci}
proves that $\psi'(p) \le 0$ in the interval $p \in (0,0.1)$. Further, for constants $c, C>0$, 
\[
- \frac{C}{p\sqrt{\log(1/p)}}\le \psi'(p) \le - \frac{c}{p\sqrt{\log(1/p)}}.
\]
Integrating, we obtain 
\[
\int_{p_1}^{p_2}- \frac{C}{t\sqrt{\log(1/t)}} dt \le \psi(p_2) - \psi(p_1) \le \int_{p_1}^{p_2}- \frac{c}{t\sqrt{\log(1/t)}} dt.
\]
Thus, we have 
\[
2C \left(\sqrt{\log(1/p_2)}- \sqrt{\log(1/p_1)}\right) \le \psi(p_2) -\psi(p_1) \le 2c \left(\sqrt{\log(1/p_2)}- \sqrt{\log(1/p_1)}\right). 
\]
This finishes the proof. 
\end{proof}

\begin{lemma}\label{lem:bounds-on-V}
For the function $V(p)=\psi^2(p)$, the function $V^{-1}(\cdot)$ is monotonically decreasing.~For any $a>0$ with $ V^{-1}(a)\le 0.1$ and any $\zeta\in (0,c_0]$ for some sufficiently small constant $c_0$, we have  
$$
V^{-1}(a)-V^{-1}(a+\zeta)=\Theta\big(\zeta V^{-1}(a)\big).
$$
\end{lemma}
\begin{proof}
As  $V(\cdot)$ is monotonically decreasing, so is $V^{-1}(\cdot)$. Taking the derivative of $V^{-1}(\cdot)$, we have
$$
\big(V^{-1}(x)\big)'=\frac{1}{V'(V^{-1}(x))} 
=\frac{1}{2\psi(V^{-1}(x))\cdot \psi'(V^{-1}(x))}.
$$
For any $x$ with $V^{-1}(x)\le 0.1$, 
by plugging in \Cref{eq:def-of-f} and \Cref{lemma:stitch-bostanci}, we have 
$$
\big(V^{-1}(x)\big)'=-\Theta\left(V^{-1}(x)\right).
$$ 
Given that $V^{-1}(\cdot)$ is decreasing, we have $V^{-1}(x)\le 0.1$ for all $x\in [a,a+\zeta]$ and thus,
$$
V^{-1}(a)-V^{-1}(a+\zeta)\le \zeta \cdot \Theta\left(  {V^{-1}(a)}\right).
$$
For the other direction, when $c_0$ is sufficiently small, we have 
from the inequality above that
$$
V^{-1}(a)-V^{-1}(a+\zeta)\le \zeta\cdot \Theta\left(  {V^{-1}(a)}\right)\le c\cdot \Theta\left(  {V^{-1}(a)}\right)\le 0.1\cdot  {V^{-1}(a)}  
$$
and thus, $V^{-1}(a+\zeta)\ge 0.9\cdot V^{-1}(a)$. 
As a result, we have 
$$
V^{-1}(a)-V^{-1}(a+\zeta)\ge \zeta \cdot \Theta\left(  {V^{-1}(a+\zeta)}\right)\ge \zeta\cdot \Theta\left(V^{-1}(a )\right).
$$
This finishes the proof of the lemma.
\end{proof}

We have the following corollary of \Cref{lem:bounds-on-V}:
 
\begin{corollary}
\label{claim:v-inverse}
 {Let $0 < p,\hat{p} \leq 0.1$, with $a:=V(p)$ and 
  $\hat{a}:=V (\hat{p})$  such that $|a-\hat{a}|\le c_0$ for some sufficiently small constant $c_0$. Then there are two positive
  constants $\tau_1$ and $\tau_2$ such that 
}
\begin{enumerate}
    \item  $|\hat{p} - p| \leq  \tau_1 p|\hat{a}-a|$; and 
    
    \item If $a - \hat{a} \ge 0$, then $\hat{p} - p \geq  \tau_2p(a-\hat{a})$.  
\end{enumerate}
\end{corollary}
\begin{proof}
The only thing worth pointing out is that when $a,\hat{a}$ satisfy 
  $|a-\hat{a}|\le c_0$ for some sufficiently small constant $c_0$,
  by the same argument above, we have $0.9\hat{p}\le p\le \hat{p}/0.9$.
\end{proof} 

\subsection{Algorithmic ingredients} \label{combined-algorithmic}

We will require a variant of the algorithm from \Cref{thm:Gaussian1} which uses a \emph{fixed} noise rate. (Recall, from \Cref{eq:def-t-xi}, that in \Cref{sec:broccoli} the noise rate $t$ depends on the error parameter $\eps$ and on the volume $p=\Vol(f)$; in contrast, we now replace the dependence on $p$ with a fixed function of $n$.) This variant, called \GSAFixedNoiseTest, is given below as \Cref{alg:GSA-Fixed-Noise-Test}.  It takes in as input a ``guessed'' value $p_2$ of the true value $p := \Vol(f)$, where the ``guessed'' value is either (i) (essentially) the correct value, or else (ii) (essentially) has \emph{one-sided error} in the sense that it is (essentially) \emph{larger} than $p$.
Since we have this one-sided guarantee on the value of $p_2$, the analysis of \GSAFixedNoiseTest\ is significantly simpler than the analysis of \GSATest. 
We state and prove the performance guarantee that we require on \GSAFixedNoiseTest~below:

\begin{theorem} \label{thm:prof-de}
    Let $f: \R^n \to \zo$ be the indicator function of a measurable subset of $\R^n$ and  $p = \Vol(f)$. 
    The algorithm~\GSAFixedNoiseTest~makes $\Theta(1/\kappa^2)$ calls to $\SAMP(f)$ and $\MQ(f)$ where $\kappa$ is as in~\Cref{alg:GSA-Fixed-Noise-Test}.
    Let $\xi,\gamma>0$ be two parameters such that 
    \begin{equation}~\label{eq:lower-bound-on-gamma-intermsoft}
    \gamma \ge C \cdot  \left(t \log^5(1/p)\right)^{1/4}. 
    \end{equation}
     and 
\begin{equation} \label{eq:sofa}
    C\cdot \frac{ \xi\sqrt{t}}{\sqrt{\log(1/p)}} 
    \leq 
    \kappa 
    \leq 
    \frac{1}{C }\cdot \frac{ \gamma\sqrt{t}}{\sqrt{\log(1/p)}},      
\end{equation}
 for $t$ as in~\Cref{alg:GSA-Fixed-Noise-Test} with some  sufficiently large constant $C$.
Then \GSAFixedNoiseTest\
 on $f,\eps$ and $p_2$ has the following properties:
\begin{enumerate}

    \item If $|p_2 - p| \leq \xi p$, then $\Pr\big[\GSAFixedNoiseTest$ accepts$\big] \geq 0.9.$

    \item If $p_2 -p\ge \gamma p$, then $\Pr\big[\GSAFixedNoiseTest$ accepts$\big] \leq 0.1.$
\end{enumerate}
\end{theorem}

\begin{proof}
{\bf Proof of Item~1:}  We note that if $f$ is a halfspace with $\Vol(f)=p$, then 
\begin{equation} \label{eq:turkey}
\frac{\NS_t(f)}{\Vol[f]} \le \frac{2\sqrt{t}}{\sqrt{\pi}} \cdot
\psi(p), 
\end{equation}
where $\psi(p) =  {\varphi(\Phi^{-1}(p))}/{p}$ 
(using the upper bound in \Cref{corr:bounds-on-stability-for-halfspaces}). We now consider two cases.

First, if $p_2 \le p$, then because $\psi(\cdot)$ is monotonically decreasing (by \Cref{lemma:stitch-bostanci}), we have
\[
\frac{\NS_t(f)}{\Vol[f]} \le \frac{2\sqrt{t}}{\sqrt{\pi}} \cdot
\psi(p) \le \frac{2\sqrt{t}}{\sqrt{\pi}} \cdot
\psi(p_2). 
\]
Since Step~2 of the algorithm computes a value $\boldsymbol{\alpha}$ such that 
\[
\bigg| \balpha - \frac{\NS_t(f)}{\Vol[f]} \bigg| \le \frac{\kappa}{2}
\]
holds with probability at least $0.99$, 
it follows that with probability at least $0.99$ we have
\[
\boldsymbol{\alpha} \le \frac{\NS_t(f)}{\Vol[f]} +  {\frac{\kappa}{2}} \le \frac{2\sqrt{t}}{\sqrt{\pi}} \cdot
\frac{\varphi(\Phi^{-1}(p_2))}{p_2} +  {\frac{\kappa}{2}},
\]
so \GSAFixedNoiseTest~accepts in this case with probability at least $0.99$.

The second case is when $p \le p_2$. In this case, $p_2- p \le \xi \cdot p$. Now we have
\begin{align}
\frac{\NS_t(f)}{\Vol[f]} &\le \frac{2\sqrt{t}}{\sqrt{\pi}} \cdot
\psi(p) = \frac{2\sqrt{t}}{\sqrt{\pi}} \cdot
\psi(p_2) + \frac{2\sqrt{t}}{\sqrt{\pi}} \cdot \big(\psi(p)-
\psi(p_2)\big) \nonumber \tag{by \Cref{eq:turkey}} \\ 
&= \frac{2\sqrt{t}}{\sqrt{\pi}} \cdot
\psi(p_2) + \frac{2\sqrt{t}}{\sqrt{\pi}} \cdot \Theta\left(\sqrt{\log(1/p)} - \sqrt{\log(1/p_2)}\right) \nonumber \tag{using~\Cref{clm:bound-psi}}\\
&\le \frac{2\sqrt{t}}{\sqrt{\pi}} \cdot
\psi(p_2) + \Theta \bigg( \frac{\xi \sqrt{t}}{\sqrt{\log(1/p)}}\bigg). 
\end{align}
Recalling \Cref{eq:sofa}, for a suitable choice of constant $C$, it follows that 
\[
\frac{\NS_t(f)}{\Vol[f]} \le \frac{2\sqrt{t}}{\sqrt{\pi}} \cdot
\psi(p_2)  + \frac{\kappa}{2}. 
\]
Since Step~2 of the algorithm computes the LHS
to error 
$\pm \kappa/2$ (with probability $0.99$), it follows that in this case,
\GSAFixedNoiseTest~accepts with probability at least $0.99$.\medskip

\noindent{\bf Proof of Item~2:} We next move to the second item. Let $H$ be a halfspace with volume $p$. Then, by Borell's theorem (Item~2 of \Cref{thm:Borell}) and \Cref{corr:bounds-on-stability-for-halfspaces}, it follows that 
\begin{equation}~\label{eq:lb-nst-c}
\frac{\NS_t(f)}{\Vol[f]} \ge \frac{\NS_t(H)}{\Vol[H]}  \ge \frac{2\sqrt{t}}{\sqrt{\pi} p} \varphi(\Phi^{-1}(p))  \cdot \left(1-\Theta\big(t \log (1/p)\big)^{1/4}\right).
\end{equation}
Using~\Cref{clm:bound-psi}, it follows  that 
\begin{equation} \label{eq:aaa}
    \psi(p) - \psi(p_2) =\Theta\left(\sqrt{\log(1/p)} - \sqrt{\log(1/p_2)}\right)  = \Theta \left( \frac{\log(p_2/p)}{\sqrt{\log(1/p)} + \sqrt{\log(1/p_2)}}\right). 
    \nonumber
\end{equation}
Using the fact that $p_2 \ge p (1+\gamma)$ and $\psi(p) = \Theta(\sqrt{\log(1/p)})$, 
it follows that 
\begin{equation}\label{eq:bound-on-diff-psi} \nonumber
\psi(p) - \psi(p_2) \ge \Theta \bigg( \frac{\gamma \psi(p)}{{\log(1/p)}}\bigg).
\end{equation}
Thus, we have 
\begin{equation}~\label{eq:otherbound-2}
\frac{2\sqrt{t}}{\sqrt{\pi}} \cdot \frac{\varphi(\Phi^{-1}(p_2))}{p_2} + \kappa \le \frac{2\sqrt{t}}{\sqrt{\pi}} \cdot \frac{\varphi(\Phi^{-1}(p))}{p} \cdot \left(1 - \Theta \left(\frac{\gamma}{\log(1/p)} \right) \right) + \kappa.  
\end{equation} 
Likewise, in Step~3 of the algorithm, the estimate $\boldsymbol{\alpha}$ satisfies 
$$
    \boldsymbol{\alpha} \ge \frac{\NS_t(f)}{\Vol[f]} - \frac{\kappa}{2} \ge \frac{2\sqrt{t}}{\sqrt{\pi} p} \varphi(\Phi^{-1}(p))  \cdot \left(1-\Theta\big(t \log (1/p)\big)^{1/4}\right) - \frac{\kappa}{2}.
$$ 
Now, by \Cref{eq:lower-bound-on-gamma-intermsoft}, we also have that 
\[
\frac{\gamma}{C} \ge \left(t \log^5(1/p)\right)^{1/4}, 
\]
for a sufficiently large constant $C_4$. 
Thus, we get that 
\[
\boldsymbol{\alpha}\ge \frac{2\sqrt{t}}{\sqrt{\pi} }\cdot \frac{\varphi(\Phi^{-1}(p))}{p}\cdot \left(1 -  \Theta\left(\frac{\gamma}{C \log(1/p)} \right) \right) -\frac{\kappa}{2}.
\]
Now, recalling \Cref{eq:sofa}, observe when $C$ is sufficiently large the right hand side of the above equation is at least as large as the right hand side of \eqref{eq:otherbound-2}. Thus, we have
\[
\boldsymbol{\alpha} > \frac{2\sqrt{t}}{\sqrt{\pi}} \cdot \frac{\varphi(\Phi^{-1})(p_2)}{p_2} + \kappa,
\]
and the algorithm will reject in this case. 
\end{proof}

\begin{algorithm}[t!]
\caption{A restricted tester for LTFs in Gaussian space, which only succeeds if the\\ estimate $p_2$ of $p = \Vol(f)$ that it is given has (essentially) one-sided error, i.e.~either $p_2$ is essentially $p$ or it is essentially larger than $p$.}
\label{alg:GSA-Fixed-Noise-Test}\vspace{0.15cm}
\GSAFixedNoiseTest$\pbra{\SAMP(f), \MQ(f), \eps,p_2}$:\\
\vspace{0.5em}
\textbf{Input:} $\SAMP(f)$, $\MQ(f)$, error parameter $\eps \in (0,1]$, estimate $p_2$ of $p := \Vol(f)$ where $p_2$ satisfies $0 < p_2 \leq 0.1$.
 \\[0.5em]
\textbf{Output:} ``Accept'' or ``reject'' 
\medskip 
\begin{enumerate}
    
    \item Set
    \[
        t =c_1\cdot  \frac{ \eps^8}{\log^5(1/p_{\mathrm{min}})} 
        \ \quad\text{and}
        \ \quad 
        \kappa = c_2\cdot {\frac {\eps^6} {{\log^{3}\pbra{1/\pmin}}}}. 
    \]
    where $c_1,c_2$ are sufficiently small constants.
    
    \item Run the algorithm \EstSense\
    to compute a value $\balpha$, which, by \Cref{lem:estimate-sensitivity}, is a\\ $(\pm \kappa/2)$-accurate estimate of $\NS_t(f)/\Vol[f]$ with confidence 0.99.
    
    \item If 
    \[
\balpha \leq \frac{2 \sqrt{t}}{\sqrt{\pi}} \cdot\frac{\varphi(\Phi^{-1}(p_2))}{p_2} +\kappa
\]
then output ``accept,'' otherwise output ``reject.''\vspace{-0.15cm}
\end{enumerate}
\end{algorithm}

\begin{algorithm}\vspace{0.15cm}
\CombinedTest$\pbra{\SAMP(f), \MQ(f), \eps}$:
\vspace{0.5em}\\
\textbf{Input:} $\SAMP(f)$, $\MQ(f)$, and error parameter $\eps \in (0,1]$. \\[0.5em]
\textbf{Output:} ``Accept'' or ``reject'' 

\medskip

\begin{enumerate}
    \item Set 
    \[
        m := \Theta
\pbra{
    {\frac {\sqrt{n \cdot{\log (1/\pmin)}}}{\eps^2}} + {\frac {\log^2(1/\pmin)}{\eps^4}}
    }
    \]    
    and draw $\x{1}, \ldots, \x{m}, \y{1}, \ldots, \y{m} \leftarrow \SAMP(f)$. 
    
    \item Compute 
    	\[
    	\bT := \frac{1}{m^2} \sum_{i, j = 1}^m \x{i}\cdot \y{j}
    	\]
     and let $\bp_2 := V^{-1}(\bT).$
    \item Run \GSAFixedNoiseTest$({\SAMP(f), \MQ(f), \eps,\bp_2})$ and output the same.\vspace{-0.15cm}
\end{enumerate}
\caption{A Gaussian relative-error LTF tester for functions of unknown Gaussian\\ volume (as long as it is not too small).}
\label{alg:combined-test}
\end{algorithm}

\subsection{Analysis of \Cref{alg:combined-test}: Proof of \Cref{thm:Gaussian3}}

Our main algorithm for \Cref{thm:Gaussian3} is called \CombinedTest~and is given in \Cref{alg:combined-test}.

By inspection~\Cref{alg:combined-test} makes 
\[
    \Theta\pbra{\frac{1}{\kappa^2}}+\Theta(m) = \Theta\pbra{\frac{1}{\eps^{12}}\cdot\log^6\pbra{\frac{1}{\pmin}}+\frac{\sqrt{n \cdot{\log (1/\pmin)}}}{\eps^2} + {\frac {\log^2(1/\pmin)}{\eps^4}}}
\]
calls to $\MQ(f)$ and $\SAMP(f)$.
So the overall complexity of the algorithm is as claimed.

\medskip
Now we turn to correctness. Let $f$ be the input function
with $p=\Vol(f)\le 0.01$.
We begin~by observing that the setting of $t$ in \GSAFixedNoiseTest~satisfies 
the upper bound that is required by \Cref{corr:bounds-on-stability-for-halfspaces}.
Next, recalling \Cref{eq:dlns-mean-formula} and \Cref{eq:var-calc-goal}, we have that
\[
    \Ex\sbra{\bT} = \frac{\bW^{=1}[f]}{p^2}
    \qquad\text{and}\qquad 
\Var\sbra{\bT} = O\pbra{\frac{\log^2(1/p)}{m} + \frac{n}{m^2} } =: \sigma^2.
\]
By making the hidden constant in the choice of $m$ sufficiently large, we can have 
$$
 {\sigma\le {c_3} \cdot \eps^2 \sqrt{\frac{\log (1/p)}{\log (1/\pmin)}}.}
$$
for some sufficiently small constant $c_3$.
The three constants $c_1,c_2$ (in \Cref{alg:GSA-Fixed-Noise-Test}) and $c_3$ as well as another constant $c^*$ in Case 2 are all sufficiently small but we require $c_3\ll c_2\ll c_1 \ll c^*$.

Consider the two cases that $f$ is an LTF and $f$ has relative distance at least $\eps$ from every LTF.\medskip

\noindent\textbf{Case 1: $f$ is an LTF.} By \Cref{fact:meaning-of-U} we have that $\bW^{=1}[f]=U(p)$ so recalling \Cref{eq:dlns-mean-formula}, we have $\E[\bT]={ {U(p)}/{p^2}}=V(p)$. By Chebyshev's inequality we have $\bT \in V(p)\pm 10 \sigma$ with probability at least 98/100, so we assume that indeed $|\bT - V(p)|  \leq 10 \sigma$. 
It follows from Item~1 of~\Cref{claim:v-inverse} (recall that $\tau_1$ is an absolute constant there) that
$|\bp_2-p|\le 10\tau_1 \sigma\cdot p.$
It suffices to show that $10\tau_1\sigma$ satisfies
  \Cref{eq:sofa} as $\xi$. Then this case follows directly from 
  Item 1 of \Cref{thm:prof-de}.
This follows from the choices of parameters and in particular, by making
  $c_3$ sufficiently smaller than $c_2$.\medskip

\noindent\textbf{Case 2: $\reldist(f, \calC_{\mathrm{LTF}}) \geq \eps$.} 
By \Cref{lem:our-MORS2} we have:
$$ 
p\sqrt{\E[\bT]} = \sqrt{\bW^{=1}[f]} 
\leq \sqrt{U(p)} -  \Omega\left({\frac {p^2\eps^2}{  \sqrt{U(p)}}}\right).
$$
Let $c^*$ be a sufficiently small constant such that 
\begin{align}
 \sqrt{\bW^{=1}[f]} 
&\leq \sqrt{U(p)} -   {c^*\cdot \frac {p^2\eps^2}{  \sqrt{U(p)}}},\text{~~~~~~~~i.e.} \nonumber\\
\E[\bT] &\leq \pbra{\sqrt{{\frac {U(p)}{p^2}}} - c^*\cdot  {\frac {\eps^2p} {  \sqrt{U(p)}}} }^2 =
{\frac {U(p)}{p^2}} - c^*{\eps^2} + {\frac {(c^*)^2\eps^4 p^2}{ U(p)}} \label{eq:luggage}
\end{align}
where the first equality is \Cref{eq:dlns-mean-formula}.
Similar to Case~1, with probability at least 98/100 we have that $\bT \in [\E[\bT] - 10 \sigma,
 \E[\bT] + 10 \sigma]$, so we assume that this is the case. Combining this with \Cref{eq:luggage}, we get that
\begin{align}
\bT \leq {\frac {U(p)}{p^2}} - c^* {\eps^2}  + {\frac {(c^*)^2\eps^4 p^2}{  U(p)}} + 10 \sigma \le {\frac {U(p)}{p^2}}-\frac{c^*\eps^2}{2},
\end{align}
by setting $c^*$  sufficiently small and $c_3$ sufficiently smaller. 
It follows from Item~2 of~\Cref{claim:v-inverse} (recall  $\tau_2$ is an absolute constant) that
$ \bp_2-p\ge (\tau_2c^*/2)\eps^2\cdot p.$
It suffices to show that $ (\tau_2c^*/2)\eps^2$ satisfies
  both \Cref{eq:lower-bound-on-gamma-intermsoft} and \Cref{eq:sofa} as $\gamma$. Then this case follows directly from 
  Item 2 of \Cref{thm:prof-de}.
This follows from the choices of parameters and in particular, 
   the order of constants that satisfies $c^*\gg c_1$ and $c_2\ll c^*,c_1$. \qed

%% file: sections/appendix.tex

\section{A lower bound for sample-based testers}
\label{ap:sample-only-lower-bound} 

\begin{theorem}
[Sample-based Gaussian LTF testing lower bound, if we are not given $p$]
\label{thm:sample-only-lower-bound} 
Let $A$ be any algorithm which uses only samples from an unknown measurable $: \R^n \to \{0,1\}$ (so in particular, $A$ is not given an estimate of $\Vol[f]$, nor does it have black-box oracle access to $f$).  If $A$ is a relative-error $0.1$-testing algorithm for LTFs over $N(0,I_n)$, then $A$ must make $\Omega(n)$ samples from $f$, even under the guarantee that $\Vol(f)$ is either $1/2$ or $1$.
\end{theorem}

The proof is an easy consequence of the following lower bound, which is proven in the arXiv version of   \cite{de2023testing}:

\begin{claim} \label{claim:linear-lower-bound}
Let $\Phi: \R \to (0,1)$ be the standard Gaussian CDF.
Let $A'$ be any algorithm which is given i.i.d.~samples drawn from  a distribution ${\cal D}$ over $\R^n$, which is either:

\begin{itemize}

\item Case 1:  ${\cal D}$ is $N(0,I_n)$; or

\item Case 2:  ${\cal D}$ is $N(0,I_n)|_{\{x \in \R^n: \Phi^{-1}(1/4) \leq u \cdot x 
\leq \Phi^{-1}(3/4)\}}$ for some unknown unit vector $u$

\end{itemize}
(so in Case 2, ${\cal D}$ is the standard Normal distribution conditioned on an unknown origin-centered ``slab'' of Gaussian volume 1/2).  Suppose that with probability at least 99/100, algorithm $A'$ correctly identifies whether ${\cal D}$ is from Case~1 or Case~2.  Then $A'$ must use $\Omega(n)$ draws from ${\cal D}$.
\end{claim}

Case~1 corresponds to the constant function $f(x)=1$, which is a halfspace, whereas it is easy to verify that in Case~2 the function $f(x)=\Indicator[\Phi^{-1}(1/4) \leq u \cdot x \leq \Phi^{-1}(3/4)]$ has volume 1/2 and has relative-distance at least 0.1 from every halfspace.
So any relative-error 0.1-testing algorithm $A$ can be used as the desired $A'$, and hence any such testing algorithm $A$ must make $\Omega(n)$ samples from $f$.
\qed